\title{Deterministic Hardness-of-Approximation of\\ 
Unique-SVP and GapSVP\\ 
in  \( \ell_p \) norms for \(p>2\)}
\date{Sep 2025}

\author{Yahli Hecht\thanks{School of Computer Science, Tel Aviv University. Supported by the European Research Council (ERC) under the European Unions Horizon 2020 research and innovation programme (Grant agreement No. 835152). email: \url{yahlihecht@mail.tau.ac.il }}
\and Muli Safra\thanks{School of Computer Science, Tel Aviv University. Supported by the European Research Council (ERC) under the European Unions Horizon 2020 research and innovation programme (Grant agreement No. 835152) and by the Israel Science Foundation (ISF) grant 2257/21. email: \url{safra@mail.tau.ac.il}}  }
\date{2025}

\documentclass{article}

\usepackage{booktabs}              
\usepackage[table]{xcolor}         
\usepackage{tabularx}              
\definecolor{RowGray}{gray}{0.97}  
\usepackage{booktabs,tabularx,array,xcolor}
\definecolor{RuleGray}{gray}{0.80}  
\definecolor{RowStripe}{gray}{0.98} 

\usepackage{amssymb, amsthm}
\usepackage{mathrsfs}
\usepackage{mathtools}
\usepackage{geometry}
\usepackage{enumitem}
\usepackage{hyperref}
\usepackage{tikz-cd}
\usepackage{graphicx}
\usepackage{cite}
\usepackage{booktabs}
\usepackage{fullpage}
\usepackage{placeins}
\usepackage{circuitikz}
\usetikzlibrary{calc}

\DeclareMathOperator{\ideg}{ideg}

\newcommand{\GapSVP}{\textsf{GapSVP}}

\makeatletter

\DeclareFontFamily{OMX}{MnSymbolE}{}
\DeclareSymbolFont{MnLargeSymbols}{OMX}{MnSymbolE}{m}{n}
\SetSymbolFont{MnLargeSymbols}{bold}{OMX}{MnSymbolE}{b}{n}
\DeclareFontShape{OMX}{MnSymbolE}{m}{n}{
    <-6>  MnSymbolE5
   <6-7>  MnSymbolE6
   <7-8>  MnSymbolE7
   <8-9>  MnSymbolE8
   <9-10> MnSymbolE9
  <10-12> MnSymbolE10
  <12->   MnSymbolE12
}{}
\DeclareFontShape{OMX}{MnSymbolE}{b}{n}{
    <-6>  MnSymbolE-Bold5
   <6-7>  MnSymbolE-Bold6
   <7-8>  MnSymbolE-Bold7
   <8-9>  MnSymbolE-Bold8
   <9-10> MnSymbolE-Bold9
  <10-12> MnSymbolE-Bold10
  <12->   MnSymbolE-Bold12
}{}

\let\llangle\@undefined
\let\rrangle\@undefined
\DeclareMathDelimiter{\llangle}{\mathopen}%
                     {MnLargeSymbols}{'164}{MnLargeSymbols}{'164}
\DeclareMathDelimiter{\rrangle}{\mathclose}%
                     {MnLargeSymbols}{'171}{MnLargeSymbols}{'171}
\makeatother
\newcommand{\aang}[1]{\llangle#1\rrangle}

\usepackage[utf8]{inputenc}
\usepackage[T1]{fontenc}

\usepackage{amssymb}         
\usepackage{amsfonts}        
\usepackage{mathtools}       
\usepackage{amsthm}

\DeclarePairedDelimiter\norm{\lVert}{\rVert}
\DeclarePairedDelimiter\abs{\lvert}{\rvert}
\DeclarePairedDelimiter\ang{\langle}{\rangle}

\DeclarePairedDelimiter\ps{(}{)}
\DeclarePairedDelimiter\ceil{\lceil}{\rceil}
\DeclarePairedDelimiter\floor{\lfloor}{\rfloor}

\DeclarePairedDelimiter\card{\lvert}{\rvert}

\DeclareMathOperator{\Span}{\mathsf{span}}

\usepackage{hyperref}
\hypersetup{
  colorlinks=true,
  linkcolor=blue,
  filecolor=magenta,
  urlcolor=cyan
}

\usepackage{thmtools}
\usepackage{thm-restate}
\usepackage[noabbrev,nameinlink]{cleveref}

\newtheorem{theorem}{Theorem}[section]
\newtheorem{definition}[theorem]{Definition}
\newtheorem{lemma}[theorem]{Lemma}
\newtheorem{claim}[theorem]{Claim}

\newtheorem{fact}[theorem]{Fact}

    \theoremstyle{definition}
\newtheorem{conj}[theorem]{Conjecture}

\usepackage[ruled,vlined]{algorithm2e}

\usepackage[table]{xcolor}

\usepackage{caption}

\usepackage{tikz}
\usetikzlibrary{positioning,calc}
\usetikzlibrary{arrows.meta,positioning,calc,shapes,fit,decorations.pathreplacing}
\usetikzlibrary{decorations.pathmorphing}

\usepackage{wasysym}
\usepackage{comment}
\usepackage{marvosym}

\newcommand\R{\mathbb{R}}
\newcommand\Z{\mathbb{Z}}
\newcommand\N{\mathbb{N}}
\newcommand\E{\mathbb{E}}

\renewcommand\S{\mathbb{S}}
\newcommand\T{\mathrm{T}}

\newcommand\F{\mathbb{F}}

\usepackage{dsfont}

\usepackage{accents}

\usepackage[ruled,vlined]{algorithm2e}

\newcommand{\NP}{\text{NP} }
\newcommand\supp[1]{\mbox{\tt supp}\left[ #1\right]}

\newcommand\poly{ {\mathsf{poly}} }

\newcommand\lat{ {\mathcal L} }

\newcommand\A{ { x} }

\DeclarePairedDelimiter\set{\{}{\}}

\newcommand\sett[2]{\left\{ #1 \;\middle\vert\; #2 \right\}}
\newcommand\Prob[2]{{\Pr_{#1}\left[ {#2} \right]}}
\newcommand\cProb[3]{{\Pr_{#1}\left[ \left. #3 \;\right\vert #2 \right]}}

\newcommand\half{{\frac{1}{2}}}

\newcommand\defeq{\stackrel{def}{=}}

\newcommand\eps{\varepsilon}

\usepackage{soul}

\usepackage[most]{tcolorbox}
\newtcolorbox{newtechbox}{
  colback=black!2, colframe=black!25, boxrule=0.4pt,
  left=1mm,right=1mm,top=0.6mm,bottom=0.6mm,
  sharp corners, fonttitle=\bfseries, title=What’s new in this section
}

\usepackage[createShortEnv]{proof-at-the-end}
\usepackage[square,numbers]{natbib}

\begin{document}

\newcommand{\PvP}{{\mathcal P v \mathcal P}}
\newcommand{\PP}{\mathbf{PL}}
\newcommand{\pvp}{{\mathcal P v \mathcal P}}
\renewcommand {\A}{\mathcal A}
\renewcommand{\P}{\mathcal P}

\maketitle

\newcommand{\largespace}{\phantom{a} \\ \phantom{a} \\ \phantom{a} \\}

\newtheorem*{lemma: consistency lemma}{Lemma \ref*{lemma: consistency-lemma}}

\newcommand\timelineDateBelow[2]{
  \pgfmathsetmacro \xdate {(#1)}
  \draw (\xdate,0.1) -- (\xdate,-0.1) node[below, align=center, text width = 2cm]{\scriptsize #2};
}

\newcommand\timelineDateAbove[2]{
  \pgfmathsetmacro \xdate {(#1)}
  \draw (\xdate,0.1) node[above, align=center, text width = 2cm]{#2} -- (\xdate,-0.1);
}

\renewcommand{\NP}{\text{NP}}
\newcommand{\svpp}{$\mathsf{SVP}_p\,$}
\newcommand{\usvpp}{$\mathsf{uSVP}_p\,$}

\begin{abstract}
We establish \emph{deterministic} hardness of approximation results for the Shortest Vector Problem in $\ell_p$ norm ($\mathsf{SVP}_p$) and for Unique-SVP ($\mathsf{uSVP}_p$)---namely, instances promised to have a unique shortest vector---for all $p > 2$.
Previously, no \emph{deterministic} hardness results were known, except for $\ell_\infty$.

\smallskip
For every $p > 2$, 
we prove constant-ratio hardness:
no polynomial-time algorithm approximates \svpp or \usvpp within a ratio of $\sqrt{2} - o(1)$,
assuming $\textsf{3SAT} \notin \text{DTIME}(2^{O(n^{2/3}\log n)})$, 
and, respectively,
$\textsf{Unambiguous-3SAT} \notin \text{DTIME}(2^{O(n^{2/3}\log n)})$.
\smallskip

We also show that for any $\eps > 0$ there exists $p_\eps > 2$ such that for every $p \ge p_\eps$: 
no polynomial-time algorithm approximates \svpp within a ratio of $2^{\ps{\log n}^{1-\eps}}$, 
assuming $\text{NP} \nsubseteq \text{DTIME}(n^{\ps{\log n}^\eps})$; 
and within a ratio of $n^{1/\ps{\log\log(n)}^\eps}$,
assuming $\text{NP} \nsubseteq \text{SUBEXP}$.
This improves upon [Haviv, Regev, Theory of Computing 2012], which obtained similar inapproximation ratios under randomized reductions.
We obtain analogous results for \usvpp under the assumptions $\textsf{Unambiguous-3SAT} \not\subseteq \text{DTIME}(n^{\ps{\log n}^\eps})$ and $\textsf{Unambiguous-3SAT} \not\subseteq \text{SUBEXP}$, improving the previously known $1+o(1)$ [Stephens-Davidowitz, Approx 2016].
\smallskip

Strengthening the hardness of \textsf{uSVP} at weaker approximation ratios has direct cryptographic impact. 
By the reduction of Lyubashevsky and Micciancio [Lyubashevsky, Micciancio, CRYPTO 2009], 
hardness for $\gamma$--$\mathsf{uSVP}_p$ carries over to ${\frac{1}{\gamma}}$--$\mathsf{BDD}_p$ (Bounded Distance Decoding).
Thus, understanding the hardness of \textsf{uSVP} improves worst-case guarantees for the two core problems that underpin security in lattice-based cryptography.
\end{abstract}

\section{Introduction}
A lattice \(\lat\subseteq\R^n\) is the additive group of all integer linear combinations of \(d\) linearly independent vectors. 
Given a linearly independent matrix \(M\in \R^{n \times d}\), we write
\(
\lat[M]\defeq \set{M \cdot \vec a \mid \vec a \in\Z^d}.
\)
For \(p\in[1,\infty]\), the \(\ell_p\) norm is \(\norm{x}_p=\big(\sum_i \abs{x_i}^p\big)^{1/p}\) for \(p<\infty\) and \(\norm{x}_\infty=\max_i \abs{x_i}\).
The length of the shortest nonzero vector is
\(
\lambda_1^{(p)}(\lat)\defeq \min_{v\in\lat\setminus\{0\}} \norm{v}_p,
\)
equivalently, the smallest \(r\) such that the closed \(\ell_p\)-ball \(B_p(0,r)\) contains a nonzero lattice point.
More generally, the \(k\)-th successive minimum \(\lambda_k^{(p)}(\lat)\) is the least \(r\) for which \(B_p(0,r)\) contains \(k\) linearly independent lattice vectors.

The main goal of this paper is to establish deterministic hardness of approximation results for two lattice problems, known as \textsf{SVP} and \textsf{uSVP}.
In the $\gamma$--\svpp problem ($\gamma > 1$), we are given a basis $M$ for $\lat[M]$ and a radius $r$, and the goal is to distinguish between the case of
$\lambda_1^{(p)}(\lat[M])\le r$ and the case of $\lambda_1^{(p)}(\lat[M])\ge \gamma r$.
In the $\gamma$--\usvpp problem, the YES instances are promised to satisfy
$\lambda_2^{(p)}(\lat[M])\ge \gamma\,\lambda_1^{(p)}(\lat[M])$. 
In words, the shortest vector in $\mathcal{L}[M]$ is unique---the only non-zero vectors in the lattice that are of length less than $\gamma \lambda_1(\mathcal{L}[M])$ are its multiples.
A closely related problem is $\textsf{BDD}$. 
In the $\frac{1}{\gamma}$--$\mathsf{BDD}_p$, we are given a basis $M$ and a target vector $\vec t$ such that $\text{dist}(\lat[M], \vec t) \le \frac{1}{\gamma} \lambda_1^{(p)}$, find the lattice vector closest to $\vec t$.
Formal definitions appear in the preliminaries (\Cref{sec: preliminaries}).

\paragraph{Hardness of SVP.}
The study of the computational hardness of the shortest vector problem (\textsf{SVP}) has a long history.
The first hardness result, due to van Emde Boas \cite{Emde1981},
established NP-hardness of \textsf{SVP} in the $\ell_\infty$ norm. 
Ajtai~\cite{Ajtai1998} proved NP-hardness in $\ell_2$, 
for an inapproximation ratio slightly larger than 1 and via a \emph{randomized reduction}.
Micciancio \cite{Micciancio1998, Micciancio2001} improved the inapproximability ratio within $2^{1/p} - o(1)$ for all $1 \le p < \infty$. 
For high $p$, Khot \cite{Khot2003} proved hardness of approximation to within $p^{1 - \eps}$. 
Khot \cite{Khot2005a} later achieved the hardness for every constant inapproximation ratio,
by constructing \textsf{GapSVP} instances with an additional structure and utilizing a variant of the tensor product to amplify the gap. 
Both Micciancio's and Khot's reductions are \emph{randomized}.

Allowing random quasi-polynomial reductions, Khot \cite{Khot2005a} also established hardness for a ratio of $2^{\log^{\half - \eps}n}$.
Haviv and Regev \cite{Haviv2012, Haviv2018} improved the ratio to $2^{\log^{1 - \eps}n}$.
Under the stronger yet plausible assumption 
NP$\not\subseteq \text{RSUBEXP} = \cap_{\gamma > 0} {\text{\small RTIME}(2^{n^\gamma})}$  the ratio reaches $n^{O\ps{1 / \log\log n}}$. 

\paragraph{Hardness of Unique-SVP.}
For the \emph{unique} variant,
Kumar and Sivakumar \cite{Kumar2001} first proved NP-hardness in $\ell_2$, albeit under randomized reductions. 
Khoat and Tan \cite{Khoat2008} proved NP-hardness of exact \textsf{uSVP} in $\ell_\infty$. 
In $\ell_p$, the best unconditional hardness factors remain very close to one:
Aggarwal and Dubey \cite{Aggarwal2016} proved hardness within $1+1/\mathrm{poly}(n)$, and Stephens-Davidowitz \cite{StephensDavidowitz2015} achieved $1 + O(\log n / n)$.
More recently, Jin and Xue \cite{Baolong2024} presented a fine-grained hardness result, reaching an inapproximability ratio of $1 + \eps$.
Under certain nonstandard assumptions, \cite{Bennett2023} established hardness for every constant ratio $\gamma \ge 1$.
Lyubashevsky and Micciancio \cite{Lyubashevsky2009} provided further evidence for the hardness of \textsf{uSVP}, by proving equivalence to \textsf{GapSVP} up to a small polynomial factor of $\sqrt{n/\log n}$.

\subsection{Our results}\label{sec:our-results}
All of the above mentioned hardness of approximation results for \svpp and \usvpp have relied on \emph{randomized} reductions (beyond \( \ell_\infty \)).
At the cost of restricting attention to \(p>2\), we obtain deterministic reductions that match---and sometimes substantially improve upon---existing hardness results.

For \textsf{SVP}, we prove hardness of approximation within $\sqrt{2} - o(1)$ for every \(p>2\), under deterministic sub-exponential reductions.
When \(p\) is sufficiently large, we deterministically match and strengthen the best known randomized results in the high-\(p\) regime, improving the $n^{O(1/\log \log n)}$ ratio of Haviv and Regev \cite{Haviv2018}.

\begin{theorem}[\textsf{SVP}, \({p > 2}\)]\label{thm:svp-low}
For every constant \(p>2\), deciding \textsf{GapSVP} in \( \ell_p \)  is hard to approximate within a ratio \(\sqrt{2}-o(1)\), unless
\( 3\textsf{SAT} \in
\text{DTIME}\ps*{2^{O(n^{2/3}\log n)}} \).
\end{theorem}

\begin{theorem}[\textsf{SVP}, high \(p\)]\label{thm:svp-high}
For every \( \eps>0 \) there exists 
\( p_\eps>2 \) such that for every \( p\ge p_\eps \),
\textsf{GapSVP} in \( \ell_p \) is hard to approximate within a ratio \( 2^{(\log n)^{1-\eps}} \),
unless \( \NP \subseteq\text{DTIME}\ps*{n^{(\log n)^{\eps}}} \).
Furthermore, under the stronger assumption \( \NP \nsubseteq \text{SUBEXP} \),
\textsf{GapSVP} is hard to approximate within a ratio
\( n^{1/(\log\log n)^{\eps}} \).
\end{theorem}

\begin{table}[h]
\centering
\small
\setlength{\tabcolsep}{7pt}
\renewcommand{\arraystretch}{1.2}
\rowcolors{3}{RowStripe}{white}
\begin{tabularx}{\linewidth}{@{}
  l !{\color{RuleGray}\vrule width 0.35pt}
  c !{\color{RuleGray}\vrule width 0.35pt}
  c !{\color{RuleGray}\vrule width 0.35pt}
  >{\raggedright\arraybackslash}X !{\color{RuleGray}\vrule width 0.35pt}
  >{\raggedright\arraybackslash}X
@{}}
\toprule
\textbf{Result} & \textbf{\(\ell_p\)} & \textbf{Approx.\ ratio} & \textbf{Assumption} & \textbf{Prior works} \\
\midrule
\Cref{thm:svp-low}  & \(p>2\)         & \(\sqrt{2}-o(1)\)
& \(3\textsf{SAT} \notin \text{DTIME}\!\left(2^{O(n^{2/3}\log n)}\right)\) & — \\
\Cref{thm:svp-high} & \(p\ge p_\eps\) & \(2^{(\log n)^{1-\eps}}\)
& \(\NP \not\subseteq \text{DTIME}\!\left(n^{(\log n)^{\eps}}\right)\)
& \(2^{(\log n)^{1-\eps}}\); \(\NP \not\subseteq \text{RTIME}\!\left(n^{(\log n)^{c}}\right)\);   \cite{Haviv2018} \\
\Cref{thm:svp-high} & \(p\ge p_\eps\) & \(n^{1/(\log\log n)^{\eps}}\)
& \(\NP \nsubseteq \text{SUBEXP}\)
& \(n^{O(1/(\log\log n))}\); \(\NP \nsubseteq \text{RSUBEXP}\); \cite{Haviv2018}\\
\bottomrule
\end{tabularx}
\end{table}
\medskip
Prior hardness results for \textsf{uSVP} lag far behind those for \textsf{SVP}, failing to reach even constant inapproximability ratios.
Beyond $\ell_\infty$, existing reductions are again \emph{randomized}.
Our reductions \emph{substantially} improve this picture.
For every \(p>2\), we give a \emph{deterministic} reduction within a ratio of (\(\sqrt{2}-o(1)\)).
For sufficiently large \(p\), we obtain a quasi-polynomial reduction showing hardness of approximation within an almost-polynomial inapproximability factor.
Thus, in the \textsf{uSVP} regime we move from sub-constant factors
(best at $1 + \frac{\log n}{n}$ by Stephens-Davidowitz \cite{StephensDavidowitz2015})
to almost-polynomial factors aligning with the known picture for \textsf{SVP}.

For \textsf{uSVP}, our results are reduced from \textsf{Unambiguous-3SAT}. 
In \textsf{Unambiguous-3SAT}, the task is to distinguish \textsf{3SAT} formulas that have exactly one satisfying assignment from those that are unsatisfiable. By the Valiant–Vazirani theorem \cite{Valiant1985}, no polynomial-time algorithm decides \textsf{Unambiguous-3SAT} unless $\text{NP} \subseteq \text{RP}$.

\begin{theorem}[\textsf{uSVP}, \({p > 2}\)]\label{thm:usvp-low}
For every constant \(p>2\), \textsf{uSVP} in \( \ell_p \) is hard to approximate within a ratio \(\sqrt{2}-o(1)\),
unless \( \textsf{Unambiguous-3SAT} \in\text{DTIME}\ps*{2^{O(n^{2/3}\log n)}} \).
\end{theorem}

\begin{theorem}[\textsf{uSVP}, high \(p\)]\label{thm:usvp-high}
For any \( \eps>0 \) there exists \( p_\eps>2 \) so that for every \( p\ge p_\eps \),
\textsf{uSVP} in \( \ell_p \) is hard to approximate within a ratio \( 2^{(\log n)^{1-\eps}} \),
unless \(\textsf{Unambiguous-3SAT} \subseteq\text{DTIME}\ps*{n^{(\log n)^{\eps}}} \).
Under the stronger assumption \(\textsf{Unambiguous-3SAT} \nsubseteq \text{SUBEXP} \),
\textsf{uSVP} is hard to approximate within a ratio \( n^{1/(\log\log n)^{\eps}} \).
\end{theorem}

\begin{table}[ht]
\centering
\small
\setlength{\tabcolsep}{7pt}
\renewcommand{\arraystretch}{1.2}
\rowcolors{3}{RowStripe}{white}
\begin{tabularx}{\linewidth}{@{}
  l !{\color{RuleGray}\vrule width 0.35pt}
  c !{\color{RuleGray}\vrule width 0.35pt}
  c !{\color{RuleGray}\vrule width 0.35pt}
  >{\raggedright\arraybackslash}X !{\color{RuleGray}\vrule width 0.35pt}
  >{\raggedright\arraybackslash}X
@{}}
\toprule
\textbf{Result} & \textbf{\(p\)-range} & \textbf{Approx.\ ratio} & \textbf{Assumption} & \textbf{Prior works} \\
\midrule
\Cref{thm:usvp-low} & \(p>2\) & \(\sqrt{2}-o(1)\)
& \(\textsf{Unambiguous-3SAT} \notin \text{DTIME}\!\left(2^{O(n^{2/3}\log n)}\right)\)
& $1 + \delta$; running-time $2^{\eps n}$ (from \textsf{SVP}); \cite{Baolong2024}\\
\Cref{thm:usvp-high} & \(p\ge p_\eps\) & \(2^{(\log n)^{1-\eps}}\)
& \(\textsf{Unambiguous-3SAT} \not\subseteq \text{DTIME}\!\left(n^{(\log n)^{\eps}}\right)\)
& $1 + \frac{\log n}{n}$; $\text{NP} \not\subseteq \text{RP}$; \cite{StephensDavidowitz2015} \\
\Cref{thm:usvp-high} & \(p\ge p_\eps\) & \(n^{1/(\log\log n)^{\eps}}\)
& \(\textsf{Unambiguous-3SAT} \nsubseteq \text{SUBEXP}\)
& \ldots \\
\bottomrule
\end{tabularx}
\end{table}

\subsection{Motivation}

\paragraph{Determinism.}
Beyond the $\ell_\infty$ norm, essentially all known hardness-of-approximation results for
\textsf{SVP} and \textsf{uSVP} rely on \emph{randomized} reductions, leaving open whether
randomness is inherently necessary.  Derandomizing these reductions has therefore been a long-standing goal.
Previous efforts have built upon the randomized reductions to $\textsf{SVP}_2$,
introducing alternative gadget constructions that may be easier to derandomize.
In this line, Micciancio~\cite{Micciancio2012} obtained a reduction achieving the
inapproximability ratio of~\cite{Haviv2018} with one-sided error,
and Bennett and Peikert~\cite{Bennett2023b} explored deterministic gadgets
based on Reed--Solomon codes.

We take a different route: returning to the algebraic \textsf{PCP} framework
of~\cite{Dinur1999,Dinur2002,Dinur2003},
and adapting its lattice encodings of NP witnesses as short lattice vectors.
We generalize this machinery and make it applicable to \textsf{SVP} in $\ell_p$ for $p>2$.
This framework is modular and we expect it to yield additional
derandomized hardness results for lattice problems,
possibly extending even to the $\ell_2$ case.

\paragraph{Cryptography and hardness of approximation.}
Tightening the approximation ratio for \textsf{uSVP} has direct consequences
for both worst–case hardness and cryptography.
Lyubashevsky and Micciancio~\cite{Lyubashevsky2009} showed that, for any \(p\ge1\),
\[
\gamma\textsf{--uSVP}_p
\ \le\
\tfrac{1}{\gamma}\textsf{--BDD}_p
\ \le\
2\gamma\textsf{--uSVP}_p,
\]
so the security of lattice–based cryptosystems whose assumptions reduce to either
\textsf{BDD} or \textsf{uSVP} directly depends on their respective approximation hardness.
Examples include encryption and signature schemes based on
\textsf{LWE} \cite{Regev2009,Peikert2016},
\textsf{NTRU} \cite{Hoffstein1998},
and \textsf{SIS}/\textsf{Ajtai–Dwork}–type constructions \cite{Ajtai1997,Goldreich1996a, Regev2004,ajtai2007first}. 
The complexity of approximating \textsf{BDD} is somewhat better understood than that of \textsf{uSVP}.
Liu, Lyubashevsky and Micciancio \cite{Liu2006} established NP–hardness
of $\textsf{BDD}_p$ for small constant approximation factors (for any \(p\ge1\)),
and Bennett and Peikert~\cite{bennett2020_BDD} showed NP-hardness for $\alpha$--$\textsf{BDD}_p$, with ratios $\alpha \to \frac{1}{2}$ as $p \to \infty$, approaching the unique-decoding radius (at most a single close vector exists---as lattice vectors are at distance of at least $\lambda_1^{(p)}$).
Surpassing this barrier provides additional motivation.

\paragraph{Attacks against LWE.}
\textsf{BDD} is precisely the decoding task underlying \textsf{LWE}
(Learning With Errors) \cite{Regev2009}.
A standard attack pipeline treats \textsf{LWE} as an instance of \textsf{BDD}
and reduces it to \textsf{uSVP}, commonly via Kannan’s embedding technique \cite{Kannan1987}:
\[
\textsf{LWE}\ \longrightarrow\ \textsf{BDD}\ \longrightarrow\ \textsf{uSVP}.
\]
This pathway and connected attacks have been analyzed and optimized in many works, including \cite{Lindner2011,Chen2011,albrecht2013efficacy,Bai2016, albrecht2017revisiting}. 
Cryptographic constructions assume hardness of approximation ratios for \textsf{uSVP}
that are much larger than the regime where NP–hardness is known or believed
(see \cite{goldreich1998limits, Cai1998, Aharonov2005}, for strong evidence against it).
Nevertheless, this motivates sharper hardness of approximation results for both
\textsf{BDD} and \textsf{uSVP}.

\subsection{Outline}
The paper is organized as follows.
\Cref{sec: preliminaries} reviews standard definitions and tools.
\Cref{sec: techniques} presents core components of our constructions, combining prior work with our  modifications and gadgets.
The sub-exponential construction
for 
\Cref{thm:svp-low} and \Cref{thm:usvp-low} appears in \Cref{sec:toy reduction}, together with its soundness analysis.
In \Cref{sec: The construction}, we give our construction for \Cref{thm:svp-high} and \Cref{thm:usvp-high}.
\Cref{sec: soundness} presents its soundness analysis.
In \Cref{sec: conclusions}, we conclude with a discussion and present some open problems for future research.

\section{Preliminaries} \label{sec: preliminaries}
\newcommand{\HH}{\mathbb{H}}
\newcommand{\vc}[1]{\vec #1}

In this section, we recall standard definitions and tools for lattices and lattice problems. 
We assume familiarity with basic concepts from PCPs, lattice geometry, and probability. 
For a thorough exposition, we refer readers to \cite{Safra2022}.



\subsection{Basic Lattice Concepts}
Recall from the introduction that for a linearly independent matrix $M\in \R^{n\times d}$, the lattice $\mathcal{L}[M]$ is the image of $\Z^d$ under $M$. 
It is often convenient to describe a lattice differently---as the set of integer solutions to a homogeneous linear system.
\begin{definition}[Integer kernel]\label{def:kerZ}
For \(A \in \Z^{m\times n}\), the \emph{integer kernel} is
\[
\ker_\Z(A) \defeq \sett{ z \in \Z^n }{ A z = \mathbf 0 } ,
\]
which is a lattice in \(\Z^n\) (and hence in \(\R^n\)).
\end{definition}

One can efficiently convert $\ker_\Z(A)$ to a representation $\mathcal{L}[M]$ of the same lattice
\begin{fact}[folklore]\label{fact:homSol}
Given \(A \in \Z^{m\times n}\), one can compute in polynomial time a basis \(M \in \Z^{n\times \dim(\ker A)}\) such that
\(\ker_\Z(A) = \lat[M]\).
\end{fact}

\subsection{Shortest Vector and Gap Problems}\label{sec:computation-probs-lattices}

Let \(\lat\subseteq\R^n\) be a lattice. For any \(\ell_p\) norm with \(p\ge 1\), the \emph{successive minima} are
\[
\lambda_k^{(p)}(\lat)
\;\defeq\;
\inf\bigl\{\, r>0 \;\mid\; \dim\bigl(\Span\{\, v\in\lat \mid \norm{v}_p \le r \,\}\bigr)\,\ge k \bigr\}.
\]
In particular, \(\lambda_1^{(p)}(\lat)\) is the length (in \(\ell_p\)) of the shortest nonzero lattice vector.

\begin{definition}[\textsf{GapSVP}]\label{def:gapsvp}
Given a full-rank basis \(M\in\R^{n\times d}\) and a threshold \(C>0\), the decision problem
\(\gamma\)--\(\textsf{GapSVP}_p\) asks to distinguish between:
\begin{itemize}
  \item \textbf{YES:} \(\lambda_1^{(p)}(\lat[M]) \le C\).
  \item \textbf{NO:}  \(\lambda_1^{(p)}(\lat[M]) > \gamma(n)\cdot C\).
\end{itemize}
\end{definition}

\begin{definition}[\textsf{Unique-SVP}; decision]
Given a full-rank basis \(M\in\R^{n\times d}\) and a threshold \(C>0\), the decision problem
\(\gamma\)--\(\textsf{uSVP}_p\) asks to distinguish between:
\begin{itemize}
  \item \textbf{YES:} \(\lambda_1^{(p)}(\lat[M]) \le C\) and $\lambda_2^{(p)}(\lat[M]) \ge \gamma(n) \cdot C$.
  \item \textbf{NO:}  \(\lambda_1^{(p)}(\lat[M]) > \gamma(n) \cdot C\).
\end{itemize}
\end{definition}

Throughout the paper, we often abbreviate $\ell_p$ and $\gamma(n)$, writing simply \textsf{GapSVP} or \textsf{uSVP} when the parameters are clear from context.
We also ignore floating-point precision, as it is insignificant.

\subsection{Constraint Satisfaction Problems}
\emph{Constraint Satisfaction Problems} (CSPs) generalize problems with local consistency constraints.
An important subcase is \emph{Constraint Satisfaction Graph} (\textsf{CSG}), in which each constraint involves a pair of variables.

\begin{definition}[\emph{Constraint Satisfaction Graph}]
A \textsf{CSG} instance consists of a graph $G = (V, E)$,
a finite alphabet $\Sigma$,
and, for each edge $e\in E$, a constraint $\Phi_e \subseteq \Sigma \times \Sigma$.
\end{definition}

An assignment $c\colon V \to \Sigma$ satisfies the \textsf{CSG} instance if
\((c(u), c(v)) \in \Phi_{(u,v)}\)
for every edge $(u,v) \in E$.
The decision problem is to determine whether a satisfiable assignment exists.

\subsection{Promise-UP}\label{sec: Promise-UP}
The class UP (Unambiguous Non-deterministic Polynomial-Time) consists of decision problems solvable by a non-deterministic polynomial-time machine that has at most one accepting computation path for each input.
Formally, a language $L$ is in UP if there exists an efficient verifier $V$ such that: 
\begin{itemize}
    \item If $x \in L$,
    then there exists a \emph{unique} witness 
    $w \in \set{0, 1}^*$ such that
    $V(x, w) = \textbf{Yes}$. 
    The witness $w$ is of length polynomial in $\card{x}$. 
    \item If $x \not\in L$, 
    then for all $w \in \set{0, 1}^*$, the verifier rejects, namely, $V(x, w) = \textbf{No}$.  
\end{itemize}
The class Promise-UP is the promise-problem analogue of UP. 
A problem is in Promise-UP 
if the YES instances have a single witness (and the NO instances have none).
The difference is that some instances fall into neither the YES nor the NO cases.

An important example is \textsf{Unambiguous-3SAT}, the problem of deciding \textsf{3SAT} instances with a promise of a unique satisfying assignment.


\subsection{Finite fields}
Let $\E$ be a field. A subset $\F \subseteq \E$ is a subfield if it is closed under the operations of $\E$ and forms a field with the induced operations. We write $\E/\F$ to denote that $\E$ is an extension of $\F$. Let $q=p^n$ be a prime power. We recall two well-known facts.
\begin{enumerate}
\item There exists, up to isomorphism, a unique field of size $q$, denoted $\F_q$.
\item For every integer $m > 1$, $\F_{q^m} / \F_{q}$. 
\end{enumerate}
It is well known that the extension $\F_{q^m} / \F_q$ can be constructed in time
$\poly\ps{q^m}$.

\subsection{Low-Degree Polynomials over Finite Fields}
Let $\F$ be a finite field and consider variables $x_1,\dots,x_t$ over \(\F\).
A \emph{monomial} is a product
$x_1^{i_1}\cdots x_t^{i_t}$,
with \emph{total degree} $\deg(x_1^{i_1}\cdots x_t^{i_t}) \defeq i_1+\dots+i_t$ and \emph{individual degree} $\ideg(x_1^{i_1}\cdots x_t^{i_t}) \defeq \max\{i_1,\dots,i_t\}$. 
The total (resp.\ individual) degree of a polynomial is the maximum total (resp.\ individual) degree among its monomials.
For any integer $d\ge0$, define
\[
  \F_{\le d}[x_1,\dots,x_t]
  =\{\,f:\F^t\to\F \mid f\text{ is a polynomial with }\deg(f)\le d\},
\]
and write $\F^t_{\le d}$ when the variables are clear.
We also work with \emph{affine planes} in $\F^t$, namely\ two-dimensional affine subspaces, denoted by $\PP(\F^t)$.  For a plane $\P\in\PP(\F^t)$, let
\[
  \P_{\le d}
  =\{\,g:\P\to\F \mid \deg(g)\le d\}.
\]

\begin{fact}[Low-Degree Extension]\label{fact: lde}
Let $\HH\subseteq\F$ and $f:\HH^t\to\F$ be any function. 
There is a unique polynomial extension $f':\F^t\to\F$ satisfying
\[
  \forall x\in \HH^t \colon f(x)=f'(x)
  \quad\text{and}\quad
  \ideg(f')\le|\HH|-1.
\]
\end{fact}
\begin{proof}
Follows immediately by interpolation.
\end{proof}

\begin{lemma}[Schwartz--Zippel]\label{lem: schwartz-zippel}
If $p\in\F_{\le d}[x_1,\dots,x_t]$ is nonzero and $S\subseteq\F$, then
\[
  \Pr_{r\in S^t}\bigl[p(r)=0\bigr]\;\le\;\frac{d}{|S|}.
\]
\end{lemma}

\subsection{Plane-vs-Plane} 
\label{sec: the plane vs plane graph}
The Plane-vs-Plane test 
(introduced by Raz and Safra \cite{Raz1997})
checks whether a purported encoding of a low-degree function is consistent.

\begin{definition}[Plane-vs-Plane]\label{def:Plane-vs-Plane}
Let $\F$ be a finite field and let $d$ be a positive integer.
Given a table \(T\) that assigns to each affine plane
\(\P \in \PP(\F^t)\) a low-degree polynomial $T[\P] \in \P_{\le d}$, 
the Plane-vs-Plane test proceeds as follows:
\begin{enumerate}
    \item Pick a random affine line \( \ell \subset \F^t \).
    Sample two distinct affine planes \(\P_1 \) and \(\P_2 \) containing \( \ell \) (precisely, \( \P_1 \cap \P_2 = \ell \)).
    \item Verify that the two functions agree on \( \ell \), namely, that \( T[P_1]_{|\ell} = T[P_2]_{|\ell} \).\end{enumerate}
\end{definition}

\begin{theorem}[Plane-vs-Plane test \cite{Raz1997}]
Let \( T \colon \PP \to \F_{\le d}[x,y] \) be an assignment of one low-degree polynomial to each plane.
There exists a constant \( c > 0 \), such that for every \( \delta > 0 \), 
if the test passes with probability $\delta$,
there exist a low-degree polynomial \(g \in \F_{\le d}[x_1,\dots,x_t] \) such that:
\[
  \Pr_{\P} \left[ T[\P] = g_{|\P}\right] \ge \delta - t \cdot \left( \frac{d}{\abs{\F}} \right)^c
\]
\end{theorem}


Here we require a different notion of soundness—\emph{list-decoding soundness}.
In the basic (unique-decoding) setting, a test passes with a non-negligible probability only if the local views partially agree with a single designated global function.
In the list-decoding variant, we allow a short list of candidate global functions, and agreement with \emph{any} one of them suffices.
Although it was likely folklore that such statements follow from the same techniques, the earliest explicit formulation we know in this context is due to Moshkovitz and Raz \cite{Moshkovitz2010a}, who also give a general recipe for deriving list-decoding guarantees from unique-decoding ones:

\begin{theorem}[Plane-vs-Plane \cite{Raz1997}: list-decoding]\label{thm:Plane-vs-Plane}
Let \( T \colon \PP \to \F_{\le d}[x,y] \) be an assignment of one low-degree polynomial to each plane.
There exists a constant \( c > 0 \), such that for every \( \delta > 0 \), there exists \( k = O\left( \frac{1}{\delta} \right) \) and a list of low-degree polynomials \( g_1, \dots, g_k \in \F_{\le d}[x_1,\dots,x_t] \) such that:
\[
  \Pr_{P_1 \cap P_2 = \ell} \left[ T[\P_1]_{|\ell} = T[\P_2]_{|\ell} \land \not\exists i\colon
  (T[\P_1] = {g_i}_{|\P_1} \land T[P_2] = {g_i}_{|P_2}) \right] \le \delta + t \cdot \left( \frac{d}{\abs{\F}} \right)^c
\]
In words: except with probability at most \( \delta + t \cdot (d/\abs{\F})^c \), every passing test (two planes agreeing on their intersection line) is explained by one of the \(k\) global low-degree polynomials on the list.
\end{theorem}

\subsubsection{Plane-vs-Plane, the Graph}

It is often helpful to think of this test as sampling an edge in the appropriate graph whose vertices consist of all planes.
We define the \emph{Plane-vs-Plane graph}
$G_{\pvp} = (V,E)$ over $\F_q^t$:
\begin{enumerate}
    \item Vertices: all affine 2-dimensional subspaces (planes) in $\F_q^t$.

    \item Edges: an (undirected) edge connects two planes if their intersection is an affine line.
\end{enumerate}
Observe that the $\pvp$ graph is regular, that is,
every vertex has the same degree. 

For any subset $S\subseteq V$, 
we define the  \emph{edge expansion} of \(S\)
as
\[
\Phi(S)\defeq
\cProb{(u, v)\in E}{u\in S}{v \not\in S}
\]
that is, the probability a random edge coming out of $S$ escapes to $V \setminus S$.

The proof of the 2-to-2 Games Theorem \cite{Khot2017,Dinur2017a,Dinur2018a,Khot2023} was completed through an analysis of set expansion in the Grassmann graph.
These expansion properties have since led to several applications, including improved low-degree testing \cite{Kaufman2025} and new PCP theorems \cite{Minzer2024_optrade}.
For our purposes, weaker expansion results suffice (see \hyperref[Appendix: A]{Appendix A} for a proof):

\begin{fact} \label{fact: expand-PvP}
Let $G_{\pvp}$ be the Plane-vs-Plane graph over $\F_q^t$.  Then for every subset $S\subseteq V$, we have
\[\Phi(S)\;
\ge\;1 \;-\;\frac{\abs{S}}{\abs{V}}
\;-\;\frac{3}{q}.
\]
\end{fact}

\section{Auxiliary Techniques} \label{sec: techniques}
In this section we present the techniques that are \emph{special} to our reductions to \textsf{SVP} in the \( \ell_p \) norm. 
Some ingredients are adapted from the PCP toolbox---in their algebraic version \cite{Dinur2011}, introduced for lattices in \cite{Dinur1998, Dinur2003}.
We tailor them to the \( \ell_p \) regime.
We present these components before formally describing the reductions, 
in order to clarify their role by separating the core ideas from the complications of the full reductions.

\subsection{Super-Assignments}\label{sec:superassign}
To construct \textsf{SVP} instances, 
we encode \textsf{CSP} via \emph{super-assignments}, following the framework of \cite{Dinur2003}.

\begin{definition}[Super-Assignment]
Let $x_1,\dots,x_n$ be \textsf{CSP} variables taking values in the alphabet 
\(\Sigma\defeq
\{\alpha_1,\dots,\alpha_k\}\).
A \emph{super-assignment}
is an integer vector with an entry for each pair
$(x_i, \alpha_j) \in \set{x_1, \dots, x_n} \times \Sigma$, namely,
\[(x_i, \alpha_j) \mapsto v_{x_i,\alpha_j} \in \Z\]
\end{definition}

\begin{definition}[Natural Assignment]
Let $g\colon \{x_1, \dots, x_n\} \to \Sigma$ be an assignment to a \textsf{CSP}.
The corresponding \emph{natural assignment} is the super-assignment defined by:  
\[
v_{x_i, \alpha_j} = 
\begin{cases}  
1 &\text{if } \alpha_j = g(x_i), \\
0 &\text{otherwise},
\end{cases}
\]
this particular super-assignment is
denoted $\ang{g}$. 
\end{definition}

Super-assignments are integer vectors, so one may impose homogeneous linear constraints on them to define a lattice.
A core technique in \cite{Dinur2003} is using a low-degree test:
Consider the Plane-vs-Plane test.
For each plane $\P \in \PP(\F^t)$ and a low-degree polynomial $g \in \P_{\le d}$ (potentially assigned to that plane),
we introduce a variable $\A_{\P}[g] \in \Z$.
Presuming a \emph{global} degree-\(d\) polynomial 
$g \in \F_{\le d}[x_1, \dots, x_t]$,
the resulting natural assignment is as follows:
\[
\ang{g}[\P, f] =
\begin{cases}
1 &\text{if } f = g|_{\P}, \\
0 &\text{otherwise}.
\end{cases}
\]
In our reductions, we will use similar variables.
Using the Plane-vs-Plane test,
we will describe homogeneous linear constraints that ensure all low-norm assignments are of the form $\alpha_1 \ang{ g_1 } + \dots + \alpha_k \ang{ g_k }$, for a small $k$.
Additional constraints will encode the clauses of the original \textsf{CSP} instances.

\subsection[The Advantage of llp Norms]{The Advantage of $\ell_p$ Norms}
After presenting super-assignments and the $\PvP$ graph, we now describe how they interact and why high-index $\ell_p$ norms are essential. 

\begin{definition}
Let $\A$ be a super-assignment for the $\PvP$ graph over $\F^t$. 
For each plane $\P \in \PP$,
define its \emph{support} as the set of functions that are assigned a nonzero value:
\[
\supp{\P}_\A \coloneqq \{f \in \P_{\le d} \mid \A_\P[f] \neq 0\}.
\]
When clear from context, we omit the subscript $\A$ and write $\supp{\P}$.
\end{definition}
\noindent
For natural assignments, we have $|\supp{\P}_\A| = 1$ for every $\P \in \PP$.

The size of the support is of interest because we can characterize $\pvp$ super-assignment whose support is bounded on every plane. 
More concretely, in \Cref{sec: pvpStruct} we prove that there exists a small constant $\eps > 0$ such that the only super-assignments satisfying a certain set of linear constraints, for which $\abs{\supp{\P}} \le |\F|^\eps$ holds on all planes, are of the form
 \(\A = a_1 \cdot \ang{ g_1 } + \dots + a_k \cdot \ang{ g_k },\)
with $k \le |\F|^\eps$.
This characterization plays a key role in the soundness analysis for \Cref{thm:svp-high} and \Cref{thm:usvp-high}

\subsubsection{Rotation}
To ensure that short vectors in the lattice, namely super-assignments with small norm, must have small supports, we add rotations.
This technique was also used in~\cite{Dinur2002} for the $\ell_\infty$ norm. 

\begin{fact} 
\label{fact: norm-tradeoff}
For every $x \in \R^m$ and $p \ge 2$, we have
\[
m^{\frac{1}{2} - \frac{1}{p}} \cdot \norm{x}_p \ge \norm{x}_2.
\]
\end{fact}
\begin{proof}
Hölder's inequality states $|\langle x, y\rangle| \le \norm{x}_p\norm{y}_q$ for $\frac{1}{p} + \frac{1}{q} = 1$.
Applying the inequality on $(x_1^2, \dots, x_m^2)$ and $(1, \dots, 1)$ with $\frac{2}{p} + \frac{p - 2}{p} = 1$:
\[\norm{x}_2^2 = \langle \vec 1, (x_1^2, \dots, x_m^2)\rangle \le m^{\frac{p-2}{p}} 
 \cdot \ps{\sum(x_i^2)^{p/2}}^{2/p} = \ps{m^{\half - \frac{1}{p}} \norm{x}_p}^2
\]
Taking square roots gives $m^{\frac{1}{2} - \frac{1}{p}} \cdot \norm{x}_p \ge \norm{x}_2$.
\end{proof}

Let $\A$ be a super-assignment and fix a plane $\P \in \PP$. 
Rotation preserves the $\ell_2$ norm, so applying any rotation matrix $U$ gives:
\[
m^{\frac{1}{2} - \frac{1}{p}} \cdot \|U \cdot \A_\P\|_p \ge \|U \cdot \A_\P\|_2 = \|\A_\P\|_2 \ge \sqrt{|\supp{\P}|}.
\]
The key point is that we can construct a specific $U$ that minimizes the $\ell_p$ norm of natural assignments. 

\paragraph{Constructing the Rotation.}
Recall the recursive definition of the Hadamard matrix: $H_1 = (1)$, and for $n > 1$,
\[
H_n = \begin{pmatrix}
H_{n-1} & H_{n-1} \\
H_{n-1} & -H_{n-1}
\end{pmatrix}.
\]
We construct our rotation matrix from a normalized Hadamard matrix, removing some columns to accommodate general dimensions (not necessarily powers of 2). 
\begin{definition}
\label{def:spinning-matrix}
For any $n \in \N$, let $m = 2^{\lceil \log n \rceil}$.
Define the rotation matrix $U \in \R^{m \times n}$ to be a normalized submatrix of the Hadamard matrix:
\[
U_{i,j} \coloneqq \frac{1}{\sqrt{m}} H_{\log m}[i,j].
\]
\end{definition}

From now on, $U$ denotes \Cref{def:spinning-matrix} with proper dimensions. 
For each standard basis vector $e_i$, we have $\|U \cdot e_i\|_p = m^{\frac{1}{p} - \frac{1}{2}}$, attaining equality in \Cref{fact: norm-tradeoff}. 
Note that if $\A$ is a natural assignment, each $\P \in \PP$ is assigned a unit vector. 
The choice of $U$ guarantees that for such $\mathcal{A}$:
\[
m^{\frac{1}{2} - \frac{1}{p}} \cdot \|U \cdot \A_\P\|_p = \|U \cdot \A_\P\|_2 = \|\A_\P\|_2 = 1.
\]

\newcommand{\Parp}{\mathrm{Par}(\F^3)}
\newcommand{\PSAT}{\mathbf{P_{sat}}}
\newcommand{\stsub}{\mathcal{R}}

\subsection{Composition-Recursion} \label{sec: Composition-Recursion}
We use the Composition-Recursion framework \cite{Arora1998} to obtain quasi-polynomial reductions.
Specifically, an algebraic version of Composition-Recursion, with modifications tailored to our setting.
This algebraic Composition-Recursion originated in \cite{Dinur2011}, to prove low-error PCP (Probabilistically Checkable Proofs) theorems. 
\cite{Dinur2003} utilized it to prove hardness of approximation for the Closest Vector Problem (\textsf{CVP}), showing inapproximability within a ratio of $n^{\frac{c_p}{\log\log n}}$,
where $c_p$ is a constant depending only on $p$.\\
Informally, the Composition-Recursion of \cite{Dinur2011} consists of two alternating steps:
\begin{enumerate}
    \item \emph{Encode} a low-degree polynomial over $\F^t$ via its restrictions to subspaces of fixed dimension. Herein, the restrictions are all to affine planes (2-dimensional subspaces). 
    \item \emph{Embed} these subspaces into a higher-dimensional vector space (typically $\F^t$), which substantially reduces the polynomial's degree.
\end{enumerate}
This process continues until the degree reaches a threshold.

Contrary to \textsf{CVP}, for \textsf{SVP} we can enforce only homogeneous linear constraints.
This degrades the soundness under Composition-Recursion.
To overcome it,
we introduce a new step to the Composition-Recursion:

\begin{enumerate}
    \setcounter{enumi}{2}
    \item \emph{Extend} the domain of the low-degree polynomial, from $\F_q^t$ to $\F_{q^2}^t$. 
    Since $\F_q^t \subseteq \F_{q^2}^t$, the polynomial can be naturally extended to $\F_{q^2}^t$ by interpreting its coefficients over the larger field $\F_{q^2}^t$.
\end{enumerate}

Field extensions are the \emph{key new ingredient} in our Composition-Recursion.
By cautiously enlarging the base field along the Composition-Recursion, we prevent the soundness degradation that previous frameworks suffer from.
This iterative process induces a forest structure as follows:
the trees correspond to planes $\P_1 \in \PP(\F_q^t)$. 
The children of each plane correspond to planes $\P_2 \in \PP(\F_{q^2}^t) $, and so on.
For convenience, we identify each vertex by the path from its root---a tuple $(\P_1, \dots, \P_r)$
where
$\P_1 \subseteq \F_{q}^t, \P_2 \subseteq \F_{q^2}^t, \dots, \P_r \subseteq \F_{q^{2^{r-1}}}^t$. 
\FloatBarrier
\noindent
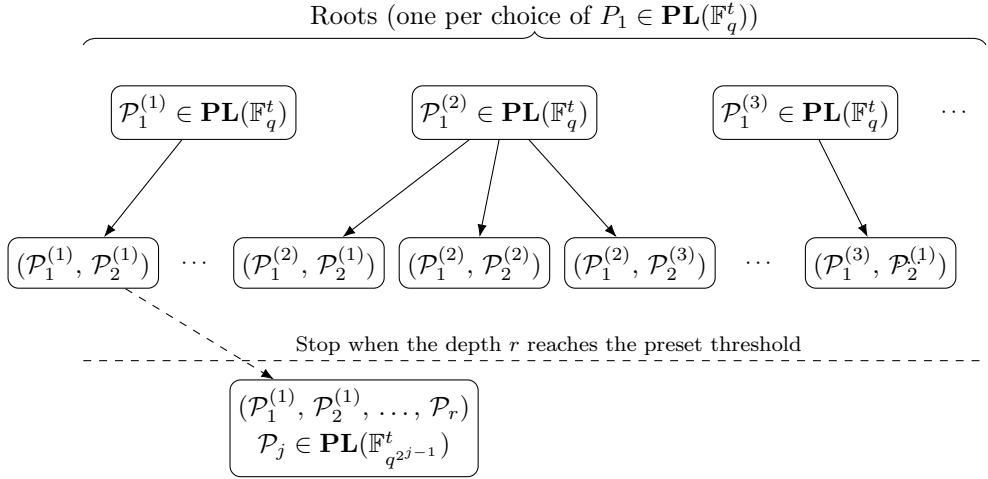
\begin{figure}[!htbp]
\centering
\begin{tikzpicture}[
  >=Latex,
  node/.style={draw, rounded corners, align=center, inner sep=3pt},
  edge/.style={-Latex},
  lab/.style={font=\footnotesize}
]
\node[node] (R1) at (-5,0) {$\P_1^{(1)}\in \PP(\F_q^t)$};
\node[node] (R2) at ( -1,0) {$\P_1^{(2)}\in \PP(\F_q^t)$};
\node[node] (R3) at (  3,0) {$\P_1^{(3)}\in \PP(\F_q^t)$};
\node[lab]  (Rdots) at (5,0) {$\cdots$};

\draw[decorate,decoration={brace,amplitude=5pt}] (-6.6,0.9) -- (5.4,0.9)
  node[midway, yshift=10pt] {Roots (one per choice of $P_1\in \PP(\F_q^t)$)};  

\node[node] (R1A) at (-6.6,-2.0) {$(\P_1^{(1)},\,\P_2^{(1)})$};
\node[lab]  (R1dots) at (-5.1,-2.0) {$\cdots$};

\node[node] (R2A) at (-3.6,-2.0) {$(\P_1^{(2)},\,\P_2^{(1)})$};
\node[node] (R2B) at (-1.4,-2.0) {$(\P_1^{(2)},\,\P_2^{(2)})$};
\node[node] (R2C) at (0.8,-2.0) {$(\P_1^{(2)},\,\P_2^{(3)})$};
\node[lab]  (R2dots) at (2.4,-2.0) {$\cdots$};

\node[node] (R3A) at (4,-2.0) {$(\P_1^{(3)},\,\P_2^{(1)})$};
\node[lab]  (R3dots) at ( 4.4,-2.0) {$\cdots$};

\foreach \S/\T in {R1/R1A, R2/R2A, R2/R2B, R2/R2C, R3/R3A}
  \draw[edge] (\S) -- (\T);

\node[node] (LeafEx) at (-3.0,-4.2) {$(\P_1^{(1)},\,\P_2^{(1)},\,\dots,\,\P_r)$\\[2pt]
$\P_j\in \PP(\F_{q^{2^{j-1}}}^t)$};
\draw[edge,dashed] (R1A) -- (LeafEx);

\draw[dashed] (-6.6,-3.3) -- (5.4,-3.3);
\node[lab,anchor=west] at (-6.2,-3.1)
  {\hspace{2.2cm} Stop when the depth $r$ reaches the preset threshold};

\end{tikzpicture}
\caption{Forest induced by Composition–Recursion. A node at depth $r$ correspond to a path $(P_1,\ldots,P_r)$ with $P_j\in \PP(\F_{q^{2^{j-1}}}^t)$. Edges are the steps in our composition (embedding, extending the field and restricting to planes). Only a few representative children are drawn; actual branching is larger.}
\label{fig:comp-rec-forest}
\end{figure}

\FloatBarrier
\noindent
The rest of this section is dedicated to surveying the mathematical definitions and facts behind embeddings and field extensions. 

\subsection{Embedding} \label{sec: embedding}
Assume arbitrary domains $\mathcal X, \mathcal Y$;
an \emph{embedding} is a structure-preserving map $E\colon \mathcal X \to \mathcal Y$.
We describe herein an embedding of $\F^k$ into $\F^{i \cdot k}$.
Let $c \in \N_{+}$, 
\[E_c\ps{(\xi_1, \dots, \xi_k)} \defeq 
(\xi_1, \xi_1^{c}, \dots, \xi_1^{c^{i-1}}, \dots, \xi_k, \dots, \xi_k^{c^{i-1}}).\]
This embedding drastically reduces the individual degree: 
\begin{fact} \label{thm: injective embedding}
    Let $f\colon \F^k \to \F$ be a polynomial of individual degree $\ideg(f) < c^i$. There exists a single polynomial $g \colon \F^{i \cdot k} \to \F$, of $\ideg(g) < c$, so that 
    \[\forall x\in\F^k\colon f(x) = g(E_c(x)).\]
\end{fact}
\begin{proof}
Let $x_1^{t_1} \cdot \dots \cdot x_k^{t_k}$ be a monomial with $0 \le t_1, \dots, t_k < c$.
For each $t_j$ there exist integers $0 \le a^j_1, \dots, a^j_{i} < c$ such that $t_j = a^j_1 + a^j_2 \cdot c + \dots + a^j_{i} \cdot c^{i-1}$.
We map the monomial to one from $\F[x_{1, 1}, \dots, x_{k, i}]$,
\[\phi(x_1^{t_1} \cdot \dots \cdot x_k^{t_k}) \defeq \prod_{m \in [i], n \in [k]} a_m^n x_{m, n}.\]
Extending $\phi$ linearly to polynomials yields an isomorphism between the polynomials of $\ideg < c^i$ in $k$ variables and those of $\ideg < c$ in $i \cdot k$ variables. 
$f((\xi_1, \dots, \xi_k)) = \phi(f)(E_c((\xi_1, \dots, \xi_k))$ so $\phi(f)$ is the unique polynomial of $\ideg < c$ satisfying the requirement.
\end{proof}

To embed an affine plane $\P \in \PP$, recall $\P = \vec \alpha + \Span(\vec x, \vec y)$, 
which defines a natural isomorphism between the polynomials over $\F^2$ and over $\P$. 
Abusing notation, we write $E(x)$ when the plane and $c$ are clear from context. 

\subsection{Field Extensions}
Let $x_1^{i_1} \cdot {\dots} \cdot x_t^{i_t} \colon \F_q^t \to \F_q$ be a monomial. 
It naturally extends to a monomial on $\F_{q^2}^t$. 
Similarly, a polynomial extends to $\F_{q^2}^t $ by extending each monomial. 
The soundness analysis requires ``reversing'' that extension. 

\begin{claim} \label{claim: reverse extension}
Let
$f \colon \F_{q^2}^t \to \F_{q^2}$ be a low-degree polynomial with $\deg(f) \le d$.
If 
\[\Prob{x_1, \dots, x_t  \in \F_q}{f(x_1, \dots, x_t) \in \F_q} > \frac{d}{q}\]
then all the coefficients of $f$ are members of $\F_q$. 
That is, $f_{|\F_q^t} \colon \F^t_q \to \F_q$ and has the same (individual and total) degree.  
\end{claim}
\begin{proof}
Fix $\alpha \in \F_{q^2} \setminus \F_q$.
It is not hard to verify that:
\[\F_{q^2} = \sett{a_1 + a_2 \cdot \alpha}{a_1, a_2 \in \F_q}\]
Every monomial is of the form 
$c \cdot x_1^{i_1} \cdot \dots \cdot x_t^{i_t}$, 
for some $c \in \F_{q^2}$,
and there exist $a_1, a_2 \in \F_q$ such that $c = a_1 + a_2 \alpha$,
so we write $f = g + \alpha h$, where the coefficients in every monomial of $h, g$ are elements of $\F_q$.

Observe that for every $x\in \F_q^t$, the values $h(x), g(x)$ are inside $\F_q$. 
Thus, $f(x) \in \F_q \iff h(x) = 0$. 
If $f$ satisfies  
$\Prob{x \in \F_q^t}{f(x) \in \F_q} > \frac{d}{q}$,
then 
$\Prob{x \in \F_q^t}{h(x) = 0} > \frac{d}{q}$
and so the Schwartz-Zippel \Cref{lem: schwartz-zippel} implies $h = 0$.
All of $f$'s monomials have coefficients from $\F_q$ and the claim follows. 
\end{proof}

\section{SUBEXP Hardness} \label{sec:toy reduction}
This section presents a sub-exponential deterministic reduction to
\(\mathsf{GapSVP}\) and \textsf{uSVP}, achieving a constant inapproximability factor---proving
\Cref{thm:svp-low} and \Cref{thm:usvp-low}.

\subsection{Parameters and notation}
Throughout this section, we fix an arbitrary norm index \(p>2\).
Our starting point is a \textsf{3SAT} formula
\(\Phi = \varphi_1\land\dots\land\varphi_m\) over variables
\(x_1,\dots,x_n\); the problem is deciding its satisfiability.
The reduction outputs an instance of
$({\sqrt{2}-o(1)})$--\(\textsf{GapSVP}_{p}\).

\paragraph{Basic parameters.}
\begin{itemize}
  \item Let \(\F\) be a finite field with
        \(|\F| \ge n^{\frac{3}{1-2/p}}\).
        (Any sufficiently large field of size
        \(n^{\Theta(1)}\) works.)
  \item Choose an arbitrary subset
        \(\HH\subset\F\) of size \(|\HH|=\lceil n^{1/3}\rceil\).
        Since \(|\HH|^3\ge n\),
        we define an injective mapping from each variable \(x_i\) to a unique
        point \(y_i\in\HH^3\); fix any such mapping $x_i \to y_i$.
  \item Set \(d=|\HH|-1\).
        a \textsf{3SAT} witness is a function
        \(\{y_1,\dots,y_n\}\subseteq\HH^3 \to \set{0, 1}\).
        By \Cref{fact: lde}, any function $\HH^3 \to \F$ admits a unique extension to
        \(\F^3\to\F\) of individual degree at most \(d\).
\end{itemize}
\paragraph{Collections of planes.}
To achieve a norm index $p$ close to $2$ we will consider a sub-collection of all affine planes in $\mathbb{F}^3$. Namely, we take:

\begin{enumerate}
  \item \emph{Parallel planes.}
        Let \(\Parp \subseteq \PP(\F^3)\) be the family of affine
        planes parallel to the coordinate axes.
  \item \emph{Clause planes.}
        For every clause involving variables \(x_i,x_j,x_k\) (regardless of their sign in the clause),
        choose an arbitrary affine plane
        \(\mathcal P_{y_i,y_j,y_k}\in\PP(\F^3)\) that contains the
        three points \(y_i,y_j,y_k\).
        Note that unless the points are collinear, this plane is unique.
        Denote
        \[
          \PSAT
          \defeq
          \bigl\{
            \mathcal P_{y_i,y_j,y_k}
            \,\bigm|\,
            \text{the clause mentions }x_i,x_j,x_k
          \bigr\}.
        \]
\end{enumerate}

For convenience, set
\(
  \mathcal R\defeq
  \PSAT \cup \Parp,
\)
a union that will recur frequently below.

\subsection{The Intermediate Lattice}
As a first step in the reduction we construct an intermediate lattice $\lat[M_I]$, 
by describing a system of homogeneous linear equations. 
By \Cref{fact:homSol}, the space of integer solutions of such a system spans a lattice $\lat[M_I]$. 
The system includes a variable $\A_\P[f]$ for every plane $\P \in \stsub$ and every low-degree polynomial $f\in \P_{\le 3d}$.
Every table $\A$ must satisfy three types of linear constraints: 
\paragraph{Simple.} 
We require that the sum of the super-assignment on every plane $\P\in\stsub$ is the same.
Equivalently, there exists a fixed global constant $\kappa\in\Z$, such that
\begin{equation*}
\sum_{f\in {\P}_{\le 3d}} A_\P[f] = \kappa \quad\text{for all } \mathcal P \in \mathcal R.
\end{equation*}
To do so, fix any plane $\P_1\in\stsub$, and for every other plane 
$\P_2 \in\stsub$ add the homogeneous equation
\begin{equation} \label{req1}
\sum_{f\in {\P_1}_{\le 3d}} A_{\P_1}[f] = \sum_{f\in {\P_2}_{\le 3d}} A_{\P_2}[f].
\end{equation}
\paragraph{Testing consistency.}
In addition, we add constraints that enforce consistency between adjacent planes.
For every two planes $\P_1, \P_2 \in\stsub$ with a nonempty intersection, and a point $x \in \P_1 \cap \P_2$,  
add the homogeneous equations 
\begin{equation} \label{req2}
\forall a \in \F \colon \sum_{f(x) = a} \A_{\P_1}[f] = \sum_{f(x) = a} \A_{\P_2}[f].
\end{equation}

\paragraph{3SAT Constraints.}
To enforce clause satisfaction, we introduce constraints on the variables associated with the planes in $\PSAT$.  
For each clause $\varphi_i$ over variables $x_{\alpha_i}, x_{\beta_i}, x_{\gamma_i}$, consider the associated plane $\P_{y_{\alpha_i}, y_{\beta_i}, y_{\gamma_i}} \in \PSAT$.  
Let $f \colon \P_{y_{\alpha_i}, y_{\beta_i}, y_{\gamma_i}} \to \F$ be a polynomial of total degree at most $3d$.  
We eliminate any $f$ that either does not represent a boolean assignment, or fails to satisfy the clause $\varphi_i$.  
This is done by imposing the constraint $\A_{\P}[f] = 0$ whenever:

\begin{enumerate}
    \item At least one of the values $f(y_{\alpha_i}), f(y_{\beta_i}), f(y_{\gamma_i})$ is not in $\{0,1\}$, or
    \item The assignment $(x_{\alpha_i} = f(y_{\alpha_i}),\; x_{\beta_i} = f(y_{\beta_i}),\; x_{\gamma_i} = f(y_{\gamma_i}))$ does not satisfy the clause $\varphi_i$.
\end{enumerate}

\subsection{The \textsf{GapSVP} instance}
Consider the intermediate lattice $\lat[M_I]$,
whose vectors correspond to tables $\A$ satisfying all constraints from the previous subsection. 
To amplify soundness, we multiply by a unitary matrix $\Tilde{U}$ and define our final lattice as
\[M_F\defeq {\left(|\stsub|\right)}^{-\frac{1}{p}}\cdot \Tilde{U} \cdot M_I.\]
The output of this reduction is the \textsf{GapSVP} instance $(M_F, 1)$. 

The matrix $\Tilde{U}$ is constructed by placing copies of the rotation matrix $U$ (from \Cref{def:spinning-matrix}) along the diagonal and normalizing.
Namely, 
if $m\defeq 2^{\ceil{\log |\P_{\le 3d}|}}$ is the number of rows in each copy of $U$, the matrix $\Tilde{U}$ is defined by
\[ \Tilde{U} \defeq {m^{\half - \frac{1}{p}}} \cdot \text{Diag}(U, \dots, U).\]
In words, we apply $U$ to each $\A_\P$ and rescale by $m^{\half - \frac{1}{p}}$. 

\paragraph{Running time.}
The number of planes is
\(|\stsub| = n^{\Theta(1)}\),
and each plane contributes
\(|\P_{\le 3d}| = |\F|^{\binom{3d+2}{2}}
   = 2^{O\!\bigl(n^{2/3}\log n\bigr)}\)
variables.
Hence $M_I$
is of size $2^{O\!\bigl(n^{2/3}\log n\bigr)}$.
Multiplying by the $\Tilde{U}$ increases the dimensions
by at most a factor of 2, which does not change the asymptotics.
Since the running time is polynomial in the lattice dimension, the
reduction's time complexity is
\(2^{O\!\bigl(n^{2/3}\log n\bigr)}\).

\paragraph{Completeness.}
Let $\sigma \colon \set{x_1, \dots, x_n}\rightarrow \set{0,1}$ be a satisfying assignment.
Our global polynomial is the low-degree extension of $x_u \rightarrow \sigma(u)$, 
guaranteed to exist by \Cref{fact: lde}. 
We denote it $g\colon \F^3 \rightarrow \F$ (having $\ideg(g) \le d$). \\

Consider the natural assignment $\ang{g}$.
Namely, for $\P \in\stsub$ and $f \in \P_{\le 3d}$, we assign $\ang{g}[\P, f] = 1$ if $g_{|\P} = f$ and otherwise $\ang{g}[\P, f] = 0$. 
It can be seen that $\ang{g}$ satisfies the \emph{simple}, \emph{consistency}, and \emph{\textsf{3SAT}} constraints. 
Thus, it is a vector in the intermediate lattice. 
For every $i\le k$, it holds 
$\norm{\frac{m^{\half - \frac{1}{p}}}{{|\stsub|}^{\frac{1}{p}}} \cdot U \cdot e_i}_p = \frac{1}{{|\stsub|}^{\frac{1}{p}}}$,
Hence, the vector $\frac{m^{\half - \frac{1}{p}}}{{|\stsub|}^{\frac{1}{p}}} \text{Diag}(U, \dots, U) \cdot \ang{g}$ has an $\ell_{p}$ norm of exactly $1$. 

\subsection{Soundness Analysis}
Now, we prove that if $\Phi$ is an unsatisfiable \textsf{3SAT} instance, then $\lambda_1^{(p)} \ge \sqrt{2} - \eta$ for every $\eta > 0$.
Let $\eta > 0$ be an arbitrary constant, and assume, by way of contradiction, that 
\[\lambda_1^{p}\left[\lat\left[\frac{m^{\half - \frac{1}{p}}}{|\stsub|^\frac{1}{p}} \text{Diag}(U, \dots, U) \cdot M\right]\right] < \sqrt{2} - \eta.\]
and fix a short nonzero vector $\frac{m^{\half - \frac{1}{p}}}{|\stsub|^\frac{1}{p}} \text{Diag}(U, \dots, U) \cdot \A$.
Denote $S\subseteq \stsub$, the subset of ``bad planes'':
\[S \defeq \set{\P \in \stsub \mid \forall i: \A_\P \neq \pm {e}_i \text{ and } \A_\P \neq \underline{0}}.\]

The soundness argument proceeds in several steps. 
We begin by establishing that $S$, the set
of ``bad planes'', is nonempty.
Specifically, we show that if no plane is bad, then the underlying \textsf{3SAT} formula is satisfiable. 
Next, we prove that $S$ contains almost all the planes. 
We apply the Schwartz–Zippel \Cref{lem: schwartz-zippel} to show that nearly all neighbors of bad planes are themselves bad. 
Using the structure of the underlying graph, this property propagates, implying that the fraction of bad planes is $1 - o(1)$. 
Finally, we observe that planes in $S$ have an increased norm. 
This leads to a contradiction, completing the argument.

\begin{lemma}
    $S$ is nonempty. 
\end{lemma}

\begin{proof}
Recall that $\A$ satisfies the ``simple'' constraints \eqref{req1}, and therefore, there is a global constant $\kappa\in\Z$ such that all the planes satisfy $\sum_{g \in \P_{\le3d}} \A_\P[g] = \kappa$. 
Now, split into cases according to the value of $\kappa$:
\begin{itemize}
    \item If $|\kappa| \neq 1$, $S$ contains all  planes with nonzero assignment.
    Because $\A \neq \vec 0$, for at least one $\P \in \stsub$, $\A_\P \neq \vec 0$ implying $\P \in S$. 
    \item If $\abs{\kappa} = 1$, we may consider $-\A$ instead, so w.l.o.g, $\kappa = 1$.
    Assume, by way  of contradiction, that $S$ is empty and
    let us construct a satisfying assignment to the \textsf{3SAT} formula as follows: 
    \begin{enumerate}
        \item For a variable $x_i$, choose an arbitrary plane that contains its corresponding point $\P \ni y_i$. 
        \item $\A_\P$ is a unit vector (since $S=\emptyset$); thus $\A_\P[g] = 1$ for a single $g$. We set $\sigma(x_i) \defeq g(y_i)$. 
    \end{enumerate}
    The \emph{consistency constraints} \eqref{req2} imply that the assignment does not depend on the selection of $\P$. 
    From the \emph{\textsf{3SAT}'s constraints}, the assignment satisfies the \textsf{3SAT} formula. \qedhere
\end{itemize}
\end{proof}

We now present the main argument
showing that $S$ contains $1-o(1)$ of the planes.
The intuition is that for each $\P \in S$, either $\norm{\A_\P}_p$ is large, or the anomalies propagate to neighboring planes via the
\emph{consistency constraints} \eqref{req2}.

Fix an arbitrary plane $\P \in S$. 
Since $\A$ is a short vector, we have $\|\frac{m^{\half - \frac{1}{p}}}{|\stsub|^\frac{1}{p}} \cdot U \cdot \A_\P\|_p \le \sqrt{2}$. 
\Cref{fact: norm-tradeoff} bounds the ratio between $\ell_2$ and $\ell_p$, resulting in 
$\|\frac{1}{|\stsub|^\frac{1}{p}} \cdot U \cdot \A_\P\|_2 \le \sqrt{2}$. 
Rearranging gives $\norm{\A_\P}_2\le\sqrt{2}|\stsub|^\frac{1}{p}$.
Since $\A_\P$ is an integer vector, the number of nonzero coordinates, namely $|\supp{\P}|$, 
is at most:
\[|\supp{\P}| \le \norm{\A_\P}_2^2 \le 2|\stsub|^\frac{2}{p}.\] 
Fix $\P \in S$ and let $f\in\supp{\P}$. 
For any distinct $g\in \supp{\P}$, the functions $f$ and $g$ agree on at most $3d \cdot |\F|$ points (From the Schwartz–Zippel \Cref{lem: schwartz-zippel}, since they are different and $\deg{f}, \deg{g} \le 3d$). 
Therefore, $f$ disagrees with all of $\supp{\P} \setminus \set{f}$ on a fraction of at least $\frac{\abs{\F} - 3d\cdot\abs{\supp{\P}}}{\abs{\F}}$ of the points in $\P$.
Note that 
\[|\stsub| \le |\Parp| + |\PSAT| = 3\abs{\F} + n^3 \le 4\abs{\F}.\]
Using the earlier bound on $|\supp{\P}|$, 
$f$ disagrees with all of $\supp{\P} \cap \set{f}$ on a fraction of at least
\[1 - \frac{6d|\stsub|^\frac{2}{p}}{\abs{\F}} \le 1 - \frac{24\ps{\ceil{n^{1/3}} - 1}\abs{\F}^\frac{2}{p}}{\abs{\F}} = 1 - 24\ps{\ceil{n^{1/3}} - 1}\ps{n^{\frac{3}{1-2/p}}}^{p/2-1} = 1 - o(1).\]
Suppose $|\supp{\P}| > 1$ and fix distinct $f, g \in \supp{\P}$. 
For a $1-o(1)$ fraction of the points $x\in\P$, $f$ differs from every function in $\supp{\P}\setminus \set{f}$,
and $g$ differs from every function in $\supp{\P}\setminus \set{g}$.
For every such $x\in \P$,
\[\sum_{h \in \P_{\le 3d}.h(x) = f(x)} \A_{\P}[h] = \A_{\P}[f] \neq 0\; \text{ and } \sum_{h \in \P_{\le 3d}.h(x) = g(x)} \A_{\P}[h] = \A_{\P}[g] \neq 0.\]
By the \emph{consistency constraints} \eqref{req2}, any neighboring plane $\P'\in\stsub$ containing $x$ has to satisfy the same equations.
Hence, any $\P'$ containing $x$ is also a ``bad plane''.

If $|\supp{\P}|=1$,
then $\A_\P$ has a single nonzero entry, with value different than $\pm 1$ (since $\P \in S$). 
The \emph{consistency constraints} \eqref{req2} immediately imply that for every plane $\P'\in\stsub$ intersecting $\P$, it must also hold that $\P'\in S$.\\

We established that for every plane $\P\in S$, a $1-o(1)$ fraction of the points $x\in\P$ are contained only in planes from $S$. 
Now, we show that $S$ contains almost all of the planes.

First, observe that there exists $\P \in S \cap \Parp$.
Recall: 
\[\Parp = \bigcup_{a\in\F} \Big\{ \sett{(a, y, z)}{y,z\in \F}, \sett{(x, a, z)}{x,z\in \F}, \sett{(x, y, a)}{x,y\in \F}\Big\}.\]
The induced $\pvp$ graph on $\Parp$ contains three independent sets of sizes $\abs{\F}$ (one for each axis), with all possible edges between different sets present. 
Since each ``bad plane'' has a  $1-o(1)$ fraction of its neighbors also in $S$, it follows that $(1-o(1))|\Parp|$ planes are in $S$.\\

Therefore $\abs{S} = (1 - o(1))\abs{\Parp} = (1-o(1))\abs{\stsub}$. 
By applying \Cref{fact: norm-tradeoff} again, for every $\P \in S$,
\[\norm{\frac{m^{1/2 - 1/p}}{|\stsub|^{\frac{1}{p}}} \cdot U \cdot \A_\P}_p \ge \frac{1}{|\stsub|^{\frac{1}{p}}} \cdot \norm{\A_\P}_2 \ge \frac{\sqrt{2}}{|\stsub|^{\frac{1}{p}}}.\]
The last inequality holds because $\P\in S$, so it is not a unit vector or $\vec 0$.
Since $S$ has a fractional size of $1 - o(1)$, we reach a contradiction.

\subsection{Unique-SVP} \label{sec: Unique SVP}

The same reduction also establishes \Cref{thm:usvp-low}.  
To this end, we start our reduction from $\textsf{Unambiguous-3SAT}$ instances.
Additionally, we assume the mapping $x_i \rightarrow y_i \in \HH^3$ is bijective.  
If $n$ is not a perfect cube, we may pad the formula with $o(n)$ dummy variables (constrained to be $0$).

To enforce the uniqueness of the shortest vector, we introduce additional linear constraints.
For every parallel plane $\P \in \Parp$ and function $f \in \P_{\le 3d}$, we enforce
\[
\text{If } \ideg(f) > d, \quad \text{then } \A_\P[f] = 0.
\]

We imposed restrictions on the reduction, so the soundness guarantee stays intact.
The completeness is still easily verifiable, and it remains to prove uniqueness, i.e.,
$\lambda_2^{(p)} \ge \sqrt{2} - o(1)$. 
To do so, we show that all short vectors are of the form $\frac{m^{\half - \frac{1}{p}}}{{|\stsub|}^{\frac{1}{p}}} \text{Diag}(U, \dots, U) \cdot \ang{g}$, and a single $g$ exists for each satisfying assignment.
For that purpose, fix any short vector.

The soundness analysis states that short vectors have no ``bad'' planes.  
Consequently, for every plane $\P \in \stsub$, $\A_\P$ is a unit vector---there is a unique function $f_{\P} \in \supp{\P}$.
Define a global function $G \colon \F^3 \to \F$ by selecting, for each $x\in\F^3$, a plane $\P \ni x$, and setting
\[
G(x) \defeq f_\P(x).
\]
The \emph{consistency constraints} \eqref{req2} guarantee that $G$ is well-defined (, independent of the choice of $\P$).
By definition, $\A = \langle G \rangle$ is a natural assignment.
The induced assignment $x_i \rightarrow G(y_i)$ satisfies the \textsf{3SAT} formula, as the \emph{\textsf{3SAT} constraints} are enforced.

Since we started from an instance of \textsf{Unambiguous-3SAT}, the restriction $G_{|\HH^3}$ is uniquely determined. 
\Cref{fact: lde} states that there is a unique low-degree extension of $x\in \HH^3 \to G(x)$ on $\F^3$.
It remains to show that $\ideg(G) \le d$.
The additional constraints enforce that the restriction of $G$ to each (affine) parallel plane $\P \in \Parp$ has individual degree at most $d$,  
and it is well-known that if all such restrictions have individual degree at most $d$, then $\ideg(G) \le d$, completing the proof.

\section{The construction} \label{sec: The construction}
We present a deterministic reduction from \textsf{3COL} to \GapSVP\ in \( \ell_p \).
Formally, we prove: 
\begin{theorem}[Main; reduction] \label{thm:restated main}
Let \(G=(V,E)\) be a \textsf{3COL} instance with \(n\defeq\abs{V}\).
For every even integer $t \ge 4$, prime power \(q=q(n)\ge n\), and \(p\ge p_t>2\),
there is a deterministic reduction mapping \(G\) to an instance of
\(\GapSVP_{\gamma}^{\,p}\) on a lattice of dimension
\(
n' = q^{O(\log^{1/\log(\frac{t}{2})}n)},
\)
with running time \(q^{O(\log^{1/\log(\frac{t}{2})}n)}\) and gap ratio
\(\gamma \;=\; q^{1/p_t}\).
\end{theorem}

\Cref{thm:svp-high} follows by instantiating \(q(\cdot)\) and $t$.
If we set \(q(n)\approx n\) and $t > 2^{1/\eps + 1}$, then \(n' =  n^{O(\log^{\eps} n)}\) and, as a function of the output dimension \(n'\),
the approximation ratio becomes
\(
2^{\Theta(\log n')^\frac{1}{1 + \eps}},
\)
and can be rewritten as $2^{\Theta(\log n')^{1 - \eps}}$.
If instead \(q(n)\approx2^{n^{\eps}}\), then \(n'=2^{O(n^{\eps}\log^\eps n)}\) and the ratio becomes \(n'^{\,\Theta\ps{1/\log \log n'}^\eps}\).
The theorem holds for every $\eps > 0$,
allowing us to avoid asymptotic notations by tuning $\eps$.

To prove hardness for unique instances (\Cref{thm:usvp-high}), minor adaptations are required, which we describe and analyze in \Cref{sec: Unique-high}.

\paragraph{Basic parameters.}
Let \(G=(V,E)\) be a \textsf{3COL} instance. 
Most parameters match those of \Cref{sec:toy reduction}.
\begin{itemize}
    \item Let $t\ge 4$ be a fixed \emph{even} integer. 
    Throughout our construction, affine planes will be embedded into vector spaces of dimension $t$. 
    \item Let $\F = \F_q$ denote a finite field of size $q \ge n = |V|$.
    \item Let $\HH \subseteq \F $ be an arbitrary set with cardinality $|\HH| = \ceil{n^{\frac{1}{t}}} $. 
    \item Let $p_t > 2$ be the (minimal) norm parameter used in our soundness analysis. 
    This is a fixed constant depending only on $t$;
    The value of $p_t$ is not explicitly stated and arises from assumptions throughout the soundness analysis.
    We prove soundness for $\ell_p$ norms where $p \geq p_t$.
    \item Again, we set $d \defeq (|\HH| - 1)$, the \emph{individual} degree of low-degree extending a function on $\HH^t$. 
\end{itemize}
Identically to \Cref{sec:toy reduction}, $\abs{\HH^t} \ge n$, so we injectively map each $v\in V$ to a unique $x_v\in \HH^t$. 

\subsection{The Intermediate Lattice}
Again, we define the intermediate lattice $\lat[M_I]$ as the integer solutions of a system of homogeneous linear equations. 
Variables in this system correspond to the leaf nodes of the Composition-Recursion tree (affine planes over an extension field) and the low-degree polynomials over these planes.

\subsubsection{Composition-Recursion Forest}
We begin by specifying the degree bounds used in the Composition-Recursion forest. 
Except at the leaves, these bounds are not explicitly enforced by the construction; the soundness analysis will show that, for short vectors, the super-assignment associated with each subtree consists of only a few polynomials that satisfy the stated individual-degree bounds.

For root nodes (namely, affine planes \(\P \in \PP(\F_q^t)\)), we distinguish two cases.
If \(\P\) is axis-parallel, the bound is \(d\).
Otherwise, the bound is \(t\cdot d\), since restricting a polynomial to an arbitrary affine plane preserves total degree but may increase the individual degree.
We now define the bounds recursively.
Suppose \(\P_1 \subseteq \F_q^t,\; \P_2 \subseteq \F_{q^2}^t,\; \dots,\; \P_r \subseteq \F_{q^{2^{\,r-1}}}^t\) form a path from a root downward, and let \(d'\) be the bound for \((\P_1,\dots,\P_{r-1})\).
Then the bound for \((\P_1,\dots,\P_r)\) is: 
\begin{enumerate}
    \item If \(\P_r\) is axis-parallel, bound the degree by \(\floor{(d')^{2/t}}\).
    \item Otherwise, bound the degree by \(t \cdot \floor{(d')^{2/t}}\).
\end{enumerate}
Which corresponds to the decrease in $\ideg$ when embedding (\Cref{thm: injective embedding}). 
This case distinction is necessary because \Cref{thm: injective embedding} requires controlling the \emph{individual} degree rather than the total degree.
We iterate the Composition-Recursion and field extension until the individual-degree bound drops below \(10t^2\).

At each leaf, if the path from the root is \(\P_1,\dots,\P_{r+1}\) and the current bound is \(d'\), 
we introduce a variable for every function \(f \colon \P_{r+1} \to \F_{q^{2^{r}}}\) with \(\ideg(f) \le d'\).
We denote this variable by \(\A[\P_1,\dots,\P_{r+1} \mid f]\).

\subsubsection{Local-to-global constraints}

Let \(\P_1 \in \PP(\F_q^t), \dots, \P_k \in \PP(\F_{q^{2^{k-1}}}^t)\) be a (possibly empty) sequence of planes, and let \(d'\) denote the individual-degree bound at \(\P_k\) (with \(d'=d\) when \(k=0\)).
The subtree rooted at \((\P_1,\dots,\P_k)\) is intended to encode low-degree polynomials on \(\P_k\) (or on \(\F_q^t\) if \(k=0\)).

We enforce a path-independence condition: for every \(x \in \F_{q^{2^k}}^t\), if we embed \(x\) along any chain of planes down to a leaf, the assigned value at \(x\) is independent of the particular chain.
Intuitively, the system encodes a low-degree polynomial by recursively encoding its restrictions to planes, and the constraints ensure that the evaluation at \(x\) is consistent across all plane choices.

To formalize this, fix any prefix \(\P_1,\dots,\P_k\) and extend it along a path \(\P_{k+1},\dots,\P_{r+1}\) from \(\P_k\) downward such that
\(x\in \P_{k+1},\; E^{(1)}(x)\in \P_{k+2},\; \dots,\; E^{(r-k)}(x)\in \P_{r+1}\)\footnote{Here \(E^{(j)}(x)\) denotes the image of \(x\) after \(j\) successive embeddings: \(E^{(1)}(x)\) is the image after embedding into \(\P_{k+1}\), \(E^{(2)}(x)\) after embedding into \(\P_{k+2}\), and so on.}.
For any alternative extension \(\Tilde{\P}_{k+1},\dots,\Tilde{\P}_{r'+1}\) with
\(x\in \Tilde{\P}_{k+1},\; \dots,\; E^{(r'-k)}(x)\in \Tilde{\P}_{r'+1}\),
we impose the local-to-global consistency constraint
\begin{equation}\label{eq:local-to-global}
\forall a\in \F_{q^{2^k}}:\quad
\sum_{f\big(E^{(r-k)}(x)\big)=a} \A[\P_1,\dots,\P_{r+1}\mid f]
\;=\;
\sum_{f\big(E^{(r'-k)}(x)\big)=a} \A[\P_1,\dots,\P_k,\Tilde{\P}_{k+1},\dots,\Tilde{\P}_{r'+1}\mid f].
\end{equation}
Moreover, to enforce descent to the parent subfield, if \(a\notin \F_{q^{2^{k-1}}}\) then, for any \(f\) with
\(f\big(E^{(r'-k)}(x)\big)=a\), we require
\[
\A[\P_1,\dots,\P_k,\Tilde{\P}_{k+1},\dots,\Tilde{\P}_{r'+1}\mid f]=0.
\]

\subsubsection{3COL constraints}
Let $\set{u, v} \in E$ be an edge in the \textsf{3COL} instance. 
We permit nonzero assignments only to low-degree functions with $x_u$ and $ x_v$ properly colored.
Suppose $\P_1, \dots, \P_{r+1}$ is a path to a leaf such that 
\[x_u, x_v\in \P_1, \dots, E^{(r)}(x_u), E^{(r)}(x_v) \in \P_{r+1}\]
We add the constraint $\A[\P_1, \dots, \P_{r+1} \mid f] = 0$
whenever $f(E^{(r)}(x_u)) \not\in \set{0,1,2}$ or $f(E^{(r)}(x_v)) \not\in \set{0,1,2}$. 
To enforce proper coloring, we also add the constraints $\A[\P_1, \dots, \P_{r+1} \mid f] = 0$
whenever $f(E^{(r)}(x_u)) = f(E^{(r)}(x_v))$.

\subsection{The GapSVP instance}
The final lattice constructed in \Cref{sec:toy reduction} 
is the result of multiplying $M_I$ with the rotation matrix of \Cref{def:spinning-matrix}.
Herein, we apply the same trick on the lowest layer of the recursion. 
In contrast to \Cref{sec:toy reduction}, planes are weighted according to their depth in the Composition-Recursion tree.

Rather than explicitly writing $M_F \defeq \Tilde{U} \cdot M_I$, we describe the operation of $\Tilde{U}$ on vectors in the intermediate lattice.
Denote $C_{r+1} \defeq \ps{|\PP(\F_q^t)| \cdot \dots \cdot |\PP(\F_{q^{2^r}}^t)|}^{-\frac{1}{p}}$. 
Let $\P_1, \dots, \P_{r+1}$ be a path to a leaf and $U$ a rotation matrix from \Cref{def:spinning-matrix}.
For the final lattice, we add: 
\[C_{r+1} \cdot m^{\half - \frac{1}{p}} \cdot U \cdot \A[\P_1, \dots, \P_{r+1} \mid *]\]
where $m$ is the number of rows in $U$. 
The number of columns in $U$ corresponds to the number of low-degree functions on $\P_{r+1}$, and the number of rows is the nearest larger power of 2.

\subsubsection{Size of the Reduction}
At each level of the Composition-Recursion tree, the individual degree \(d\) decreases to
$d' \;\le\; t \cdot \big\lfloor d^{\,2/t}\big\rfloor$.
The process stops once the degree is below some constant. 
Let \(R\) be the depth of the Composition-Recursion tree. 
Writing $\log(d') \le \log(t) + \frac{2}{t}\log(d)$, we obtain
\[
R \;=\; \frac{\log\log d}{\log(t/2)} \;+\; O(1).
\]

\paragraph{Number of leaves.}
At depth \(i\) the field size is \(q^{2^{i}}\), and the number of affine planes in \(\F_{q^{2^{i}}}^{\,t}\) is at most
\(|\mathrm{PL}(\F_{q^{2^{i}}}^{\,t})| \le (q^{2^{i}})^{3t}\).
The total number of leaves is bounded by
\[
\prod_{i=0}^{R+1} |\mathrm{PL}(\F_{q^{2^{i}}}^{\,t})|
\;\le\;
\prod_{i=0}^{R+1} (q^{2^{i}})^{3t}
\;=\;
q^{\,3t\,(2^{R+2}-1)}.
\]
Applying our bound on the depth of the recursion,
$
2^{R}
\;=\;
2^{\frac{\log\log d}{\log(t/2)} + O(1)}
\;=\;
O\ps{(\log d)^{1/{\log(\frac{t}{2})}}}.
$
Therefore, the number of leaves is at most
$
q^{O\ps{\log^{1/{\log(\frac{t}{2})}}(d)}}.
$

\paragraph{Per‑leaf.}
Each leaf plane stores an entry for every low‑degree function on that plane. 
At the leaves, the individual
degree is at most \(10t^2\), 
so the number of entries is at most
\[
|\F_{q^{2^{R+1}}}|^{100t^4}
\;=\;
q^{\,O(2^{R+1})} \;=\; q^{O\ps{\log^{1/{\log(\frac{t}{2})}}(d)}}.
\]
Summing over the leaves,
the number of entries---$\lat[M_I]$ lattice's dimension---is at most
$
q^{O\ps{\log^{1/{\log(\frac{t}{2})}}(d)}}.$
Note that the dimension of the final lattice $\lat[M_F]$ is at most doubled, which does not change the asymptotics.
The time complexity is polynomial in the size of the output.

\subsection{Completeness}
Fix a satisfying coloring $\sigma \colon V \rightarrow \set{0, 1, 2}$,
and let $g\colon \F^t \rightarrow \F$ be a low-degree extension of the map $x_v \rightarrow \sigma(v)$.
By \Cref{fact: lde}, the individual degree satisfies $\ideg(g) \le d$.
Restricting, embedding, and extending $g$ along the Composition-Recursion forest naturally defines a vector in $\lat[M_I]$:

\begin{enumerate}
    \item At the root of the Composition-Recursion forest, each vertex corresponds to an affine plane $\P \in \PP\ps{\F_q^t}$. 
    We restrict $g$ to that plane, obtaining $g_{|\P}$.
    \item If a vertex $(\P_1, ..., \P_r)$ is associated with a polynomial $f$, we embed $f$ according to \Cref{thm: injective embedding}, obtaining $f^*$.
    The low-degree polynomial $f^*$ extends naturally when passing to an extension field.
    The child vertex $(\P_1, ..., \P_{r}, \P_{r+1})$ is then assigned $f^*_{|\P_{r+1}}$.
    \item At a leaf $(\P_1, ..., \P_k)$ with a function $f$, we set $\A[(\P_1, ..., \P_k) \mid h] = 1_{f}(h)$, exactly as in \Cref{sec:toy reduction}.
\end{enumerate}

Denote such an assignment $\aang{g}$.
It is not hard to confirm that $\aang{g}$ satisfies the linear constraints, and thus $\aang{g} \in \lat[M_I]$. 
From now on, we also refer to $\aang{g}$ as a natural-assignment.
$\Tilde{U} \cdot \aang{g} \in \lat[M_F]$ is a lattice vector of $\ell_p$ norm exactly 1. 

\section{Soundness} \label{sec: soundness}
Fix a norm parameter \(p \ge p_t\).
Suppose \(\Tilde U \cdot \A \in \lat[M_F]\) is a lattice vector with
\(\|\Tilde U \cdot \A\|_p \le q^{1/p_t}\).
We will show that, for a sufficiently large absolute constant \(p_t>2\) (depending only on the construction parameters, e.g., \(t\)), the underlying \textsf{3COL} instance must be satisfiable.
For readability, we do not attempt to optimize \(p_t\); any fixed sufficiently large choice of \(p_t\) suffices.

Before diving into the soundness analysis, we first present the necessary analytic tools.

\begin{definition}[Weak Plane-vs-Plane constraints] These are ``local-to-global'' constraints implicitly enforced in \Cref{sec: The construction}.
Let $\A\neq \vec 0$ be a super-assignment --- a vector with an entry for each pair $\P \in \PP(\F^t)$ and $f \in \P_{\le d}$.
For every affine line $\ell = \P_1 \cap \P_2$, 
we add the equations:  
\begin{equation} \label{eq: weak local-to-global}
\forall x\in \ell, \forall a\in \F \colon \sum_{f(x) = a} \A_{\P_1}[f] = \sum_{f(x) = a} \A_{\P_2}[f].
\end{equation}
Thus, instead of enforcing consistency on the entire line $\ell$, we enforce it pointwise at each $x\in \ell$.
\end{definition}

\begin{lemma: consistency lemma}[$\pvp$ local to global]
There exists an absolute constant $\eps > 0$ such that the following holds.
Let $\A$ be a super-assignment that assigns an integer to each $\P \in \PP$ and $f \in \P_{\le d}$.
Assume $d \le |\F|^\eps$ and $|\supp{\P}| \le |\F|^\eps$ for every $\P \in \PP$.

If $\A$ satisfies the weak Plane-vs-Plane constraints, 
then there exist integers $a_1, \dots, a_k \in \Z$, $k \le |\F|^\eps$, and global degree-$d$ functions $g_1, \dots, g_k \colon \F^t \to \F$ such that 
\[
\A = a_1 \cdot \ang{ g_1 } + \dots + a_k \cdot \ang{ g_k }.
\] 
\end{lemma: consistency lemma}

\Cref{lemma: consistency-lemma} is proven in \Cref{sec: pvpStruct}.
By combining the field-extension step with high \( \ell_p \) norms, we apply the lemma at the lowest recursion level.
Working upwards along the Composition-Recursion tree, we conclude that any short vector is an integer combination of few natural assignments.
Finally, we show that the resulting low-degree polynomials encode proper 3-colorings of the original \textsf{3COL} instance.

\subsection{Characterizing short vectors}
Informally, our argument is as follows.
At each level of the recursion, the \emph{field-extension} step preserves a polynomial ratio between the number of vertices and the ambient field size.
Consequently, any attempt to bypass the local-to-global constraints within a subtree must ``spend mass'':
it increases the \(\ell_p\) norm on that subtree by a factor polynomial in the field size.
Because we work with \(p\ge p_t\) and \(p_t\) is a fixed large constant, these polynomial losses are amplified, so such \emph{local} assignments necessarily have significantly larger \(\ell_p\) norm.

In the next section, our goal is to prove an analogue of \Cref{lemma: consistency-lemma}, on the whole Composition-Recursion tree ---
that $\A$ admits a representation by a small list $(a_i, f_i)_{i \in I}$.
Formally, we aim to write $\A = \sum_{i \in I} a_i \aang{f_i}$, when $|I|$ is bounded.
$\sum_{i \in I} a_i \aang{f_i}$ will also be referred to as a \emph{super-assignment}. 

\paragraph{Consistent vertices.} 
We construct such a representation by working from the leaves upward.
A vertex $\P_1, \dots, \P_r$, not necessarily a leaf, will be called \emph{consistent} if there exists a list $(a_i, f_i)_{i\in I}$
of low-degree polynomials $\F_{q^{2^r}}^t \to \F_{q^{2^r}}$ (With the individual-degree satisfying the Composition-Recursion degree bound) such that: 
\begin{enumerate}
    \item $|I| \le |\F_{q^{2^r}}|^{\frac{12 \cdot t}{p_t}}$, and
    \item The assignment induced on the subtree agrees with $(a_i, f_i)_{i\in I}$. 
\end{enumerate}
Note that consistency propagates downward: if a vertex is consistent, then every descendant is also consistent.  
Our objective is therefore to prove that the root ($r=0$) is consistent.  
For the sake of contradiction, assume otherwise, and let $\P_1, \dots, \P_r$ be an inconsistent vertex of maximal depth.  

\subsection{The offspring's support}
For convenience, assume that $\P_1, \dots, \P_r$ is not a leaf. 
The case of leaves is simpler and follows the same steps. 

As $\P_1, \dots, \P_r$ is of maximal depth, each offspring $\P_1, \dots, \P_r, \P_{r + 1}$ is explained by a short list $I = (a_i, f_i)$. 
Using the consistency of $\P_1, \dots, \P_r, \P_{r+1}$, and that $\Tilde{U} \cdot \A$ has low norm, we improve the bound on $(a_i, f_i)$'s length. \\

Let $\P_1, \dots, \P_{r+1}, \dots \P_k$ be a leaf. 
Again, $\supp{\P_1, \dots, \P_k}$ contains the low-degree polynomials with nonzero coefficient. 
Similarly to \Cref{sec:toy reduction}, the support's size gives a lower bound on the norm:
\[\|C_{k} \cdot U \cdot \A[\P_1, \dots, \P_{k} \mid *]\|_2  \le 
\|C_{k} \cdot m^{\half - \frac{1}{p}} \cdot U \cdot \A[\P_1, \dots, \P_{k} \mid *]\|_p\]
\[C_k |\supp{\P_1, \dots, \P_k}|^{\half}  \le 
\|C_{k} \cdot m^{\half - \frac{1}{p}} \cdot U \cdot \A[\P_1, \dots, \P_{k} \mid *]\|_p\]

Returning to $\P_1, \dots, \P_{r+1}$, the Schwartz-Zippel \Cref{lem: schwartz-zippel} implies that $(f_i)$ collides on a fraction of at most $d|\F_{q^{2^{r+1}}}|^{\frac{24 \cdot t}{p_t} - 1}$ of the planes. $p_t$ is sufficiently large, making the collisions negligible. 

By definition, an offspring $\P_1, \dots, \P_{r+1}, \P$ is specified by the list $(a_i, f^*_i)$, where $f^*_i$ is obtained by embedding and extending ${f_i}_{|\P}$. 
Whenever $(f_i)$ does not collide on $\P$, the list $(a_i, f^*_i)$ contains $|I|$ distinct low-degree functions.
This property persists downwards, yielding $|I|$ distinct functions for all but an $o(1)$ fraction of the descendants of $\P_1, \dots, \P_{r + 1}$.

Consequently, the norm on the descendants of $\P_1, \dots, \P_{r+1}$ is at least: 
\[(1 - o(1)) \cdot \ps{\sum_{\P_{r+2}, \dots, \P_k}C_k^p |I|^{\frac{p}{2}} }^{\frac{1}{p}} = 
(1 - o(1)) \cdot  C_{r+1} |I|^\half
\]

where the equality follows from $C_k^p = \ps{|\PP(\F_q^t)| \cdot \dots \cdot |\PP(\F_{q^{2^{k-1}}}^t)|}^{-1}$. 
The total norm is at most $q^{\frac{1}{p_t}}$, and thus
$(1 - o(1)) \cdot  C_{r+1} |I|^\half \le q^{\frac{1}{p_t}}$. 
A lower bound for $C_{r+1}$ implies an upper bound on $I$.
Using $|\PP(\F_x^t)| \le x^{3t}$ and $p > p_t$, it holds that $C_{r+1} \ge  {q^{-\frac{3t}{p_t}}} \cdot {\dots} \cdot {q^{-\frac{3t}{p_t} \cdot 2^{r}}} = {q^{-\frac{3t}{p_t} \cdot (2^{r + 1} - 1)}}$.
we obtain: 
\[(1 - o(1)) \cdot  {q^{-\frac{3t}{p_t} \cdot (2^{r + 1} - 1)}} |I|^\half \le q^{\frac{1}{p_t}}\]
\[|I|^\half \le {q^{\frac{3t}{p_t} \cdot (2^{r + 1})}} = |\F_{q^{2^{r+1}}}|^{\frac{3 t}{p_t}}\]
Exactly a square root of the requirement on $|I|$. 
Restating, $|I| \le |\F_{q^{2^{r}}}|^\frac{12\cdot t}{p_t}$.

\subsection{Pulling up}
Every offspring $\P_1, \dots, \P_r, \P$ has a short description $( a_i^\P, f_i^\P)_{i \in I_\P}$, where $I_\P \le |\F_{q^{2^{r}}}|^\frac{12\cdot t}{p_t}$.
To describe $\P_1, \dots, \P_r$, one may be tempted to use $\Cref{lemma: consistency-lemma}$. 
Indeed, this is our end goal. 
However, $( a_i^\P, f_i^\P)$ aren't assignments to $\P$, but low-degree polynomials on 
$\F_{q^{2^{r+1}}}$ ($\P$ after embedding and extending the field).
Before applying the lemma,
it is necessary to ``reverse'' the field-extension and embedding. 

\paragraph{Reversing the Composition-Recursion.} 
We show that for every $f_i^\P$, the restriction ${f_i^\P}_|{\F_{q^{2^{r}}}^t}$ defines a function $\F_{q^{2^{r}}}^t \to \F_{q^{2^{r}}}$.
By \Cref{claim: reverse extension}, it is sufficient to prove that $f_i^\P$ output values are in $\F_{q^{2^{r}}}$ with non-negligible probability when restricted to ${\F_{q^{2^{r}}}^t}$.
If $(f_i^\P)$ does not collide on a point $x\in \F_{q^{2^{r}}}^t$, then by the local-to-global constraints for all $i\in I_\P$ we have $f_i^\P(x) \in \F_{q^{2^{r+1}}}$. 
The Schwartz-Zippel \Cref{lem: schwartz-zippel} ensures that such collisions rarely occur.
Thus, $\Cref{claim: reverse extension}$ guarantees that the restriction ${f_i^\P}_{|\F_{q^{2^{r}}}^t}$ is itself a low-degree polynomial $\F_{q^{2^{r}}}^t \to \F_{q^{2^{r}}}$.
Since $f_i^\P$ is a low-degree polynomial, it follows that $f_i^\P$ is exactly the low-degree extension of $f_i^\P |_{\F_{q^{2^{r}}}}$.\\

Using \Cref{thm: injective embedding}, for every $f_i^\P$ there exists a unique $g_i^\P : \P \to \F_{q^{2^{r}}}$ such that  embedding $g_i^\P$ yields $f_i^\P |_{\F_{q^{2^{r}}}}$, and extending the field recovers $f_i^\P$.

\paragraph{Constructing a super-assignment.}
Now, the conditions for \Cref{lemma: consistency-lemma} are satisfied.
Consider the $\pvp$ graph on $\F_{q^{2^r}}^t$. 
For every plane $\P$, we assign $a_i^\P$ to $g_i^\P$, and $0$ elsewhere.

By \Cref{eq:local-to-global}, the weak Plane-vs-Plane constraints (\ref{eq: weak local-to-global}) are satisfied. 
Applying \Cref{lemma: consistency-lemma}, we conclude that there exists a global description $(a_i, f_i)_{i \in I}$, with $|I| \le |\F_{q^{2^{r}}}|^\frac{12\cdot t}{p_t}$.\\

Cautious readers may notice we are not yet finished. 
It remains to show that the functions $(f_i)$ have low \emph{individual} degree, matching the degree bounds prescribed by the Composition-Recursion forest. 
However, \Cref{lemma: consistency-lemma} only guarantees that the (total) degree is low enough.

Recall that on \emph{parallel} planes $\P \in \PP(\F_{q^{2^{r}}})$, the individual degree $\ideg(g_i^\P)$, is sufficiently low. 
By the Schwartz-Zippel \Cref{lem: schwartz-zippel}, the functions $(f_i)$ collide on only a negligible fraction of these parallel planes. 
Thus, for every $f_i$, restriction to almost all parallel planes has low individual degree.
It then follows that $f_i$ has a low individual degree.

\subsection{Satisfying the 3COL} \label{sec: sat-col}
Having established a global description $(a_i, f_i)$, 
we now argue that every $f_i$ induces a satisfying assignment $u \to f_i(x_u)$ to the underlying \textsf{3COL} instance.

Fix any point $x\in \F_q^t$ such that $(f_i)$ does not collide on $x$, and let $\{u, v\} \in E$ be an edge.
Let $\P_1$ be an affine plane containing $x, x_u, x_v$. 
Let $\P_2$ be an affine plane containing $E(x), E(x_u), E(x_v)$, the images of these points under the embedding of $\P_1$.
We can continue recursively and obtain a path to a leaf $\P_1, \dots, \P_r$, together with corresponding points $\Tilde x, \Tilde{x_u}, \Tilde{x_v} \in \P_r$.

Since for all ${i \neq j}$, it holds that $f_i(x) \neq f_j(x)$, recursively embedding and extending $(f_i)$ 
yields a collection $(f_i^*)$ of distinct low-degree functions on $\P_r$. 
For each $f_i$ with nonzero coefficient, 
$\Phi$'s constraints ensure that $(f_i^*(\Tilde{x_u}), f_i^*(\Tilde{x_v}))$ satisfies the \textsf{3COL} constraint on $\{u,v\}$.
$f_i(x_u) = f_i^*(\Tilde{x_u})$ and $f_i(x_v) = f_i^*(\Tilde{x_v})$, thus $u \to f_i(x_u)$ is a satisfying coloring.

\subsection{Unique-SVP} \label{sec: Unique-high}
In the soundness analysis, we showed that short vectors correspond to $\sum_{i \in I} a_i \aang{f_i}$, where $\ideg(f_i) \le d$.
In order to start from $\textsf{Unambiguous-3SAT}$, we need to map a formula $\Phi = (\varphi_1 \land \dots \land \varphi_m)$ over variables $x_1, \dots, x_n$ to $\HH^t$ and replace the \textsf{3COL}'s constraints.

Fix any mapping $\set{\varphi_1, \dots, \varphi_m} \cup \set{x_1, \dots, x_n} \to \HH^t$.
Again, we assume it is bijective (we may pad the formula with $o(n)$ dummy variables constrained to be $0$). Instead of the \textsf{3COL} constraints, we enforce:
\begin{itemize}
    \item \emph{Alphabet:} Let $\varphi_i$ be a formula, $x_{\varphi_i}\in\HH$ its corresponding point, and $\Tilde{x}_{\varphi_i}$ the result of embedding $x_{\varphi_i}$ until reaching a leaf $\P_1, \dots, \P_r$. 
    We enforce $\A[\P_1, \dots, \P_{r+1} \mid f] = 0$ whenever $f(\Tilde{x}_{\varphi_i}) \notin \set{1, \dots, 7}$. 
    Similarly, let $x_i$ be a variable, $x_{x_i}\in\HH^t$ its corresponding point, and $\Tilde{x}_{x_i}$ the result of embedding. 
    We enforce $\A[\P_1, \dots, \P_{r+1} \mid f] = 0$ whenever $f(\Tilde{x}_{x_i}) \notin \set{0, 1}$.

    \item \emph{Consistency:} Let $\varphi_i$ be a formula and $x_{x_j}$ be a variable in $\varphi_i$.
    After embedding both until reaching a leaf $\P_1, \dots, \P_{r+1}$, we have $\Tilde{x}_{\varphi_i}, \Tilde{x}_{x_j} \in \P_{r + 1}$. Thinking about $f(\Tilde{x}_{\varphi_i})$ as 3 bits---each equal to 1 if the corresponding literal is satisfied, and $f(\Tilde{x}_{x_j})$ as the assignment to $x_j$; for every $f$ on $\P_{r+1}$ we enforce $\A[\P_1, \dots, \P_{r+1} \mid f] = 0$ if these bits are inconsistent.
\end{itemize}

Proving that each $x_i \to f_i(x_i)$ is a satisfying assignment to $\Phi$ follows identically to $\textsf{3COL}$ (\Cref{sec: sat-col}).
Note that starting from $\textsf{3SAT}$ instead of $\textsf{3COL}$ would have been possible, but since the constraints are somewhat harder to follow, we chose to start from $\textsf{3COL}$.

Now, we need to prove \emph{uniqueness}, i.e., $\lambda_2^{(p)} \ge q^{1/p_t} \lambda_1^{(p)}$.
Identically to \Cref{sec: Unique SVP}, \Cref{fact: lde} promises that each $f_i$ is the unique low-degree extension of the single satisfying assignment, and so up to multiplication by a scalar, there exists a single short vector.

\subsection{Characterizing short PvP super-assignments}
\label{sec: pvpStruct}
To finalize the soundness analysis, it is left to prove \Cref{lemma: consistency-lemma}.
While \cite{Dinur2003} proved a similar result, 
we present a proof based on $G_\pvp$'s expansion, that leads to an arguably cleaner reduction. 
We prioritize readability over tight parameters; thus, some constants are not presented in their optimal form.

\paragraph{Expansion of nonzero planes:} 
Before proceeding to \Cref{lemma: consistency-lemma}, 
we claim assignments with bounded support have nonzero values on almost all the planes. 
\Cref{non-zero-short} allows us to ignore planes $\P \in \PP$, whenever $\A_\P = \vec 0$. 

\begin{claim} \label{non-zero-short}
Let $\eps > 0$ and $\A\neq \vec 0$ be a vector over $\Z$, with an entry for every $\P \in \PP(\F^t)$ and $f \in \P_{\le d}$.
Suppose $\A$ satisfies weak Plane-vs-Plane (\ref{eq: weak local-to-global}) constraints, and for every plane $\P \in \mathbf P$, the support of $\P$ is bounded by $|\supp{\P}| \le |\F|^{\eps}$. 
Then, $\A_\P = \vec 0$ for a fraction of at most $(d + 3)|\F|^{-1 + \eps}$ of the planes. 
\end{claim}

\begin{proof}
Fix such $\A$ and let $S \subseteq \PP$ be the planes with nonzero assignments. 
Namely, $S \defeq \sett{\P \in \PP}{\A_\P \neq \vec 0}$. We will prove it is poorly expanding, and apply \Cref{fact: expand-PvP} to show $S$ contains almost all the planes. \\

Fix $\P \in S$ and an arbitrary function $f\in \supp{\P}$.
For any different function $g\in \supp{\P}$, 
the Schwartz-Zippel \Cref{lem: schwartz-zippel} states that $f$ and $g$ agree 
on a random point with probability at most $\frac{d}{\abs{\F}}$. 
As $\abs{\supp\P} \le |\F|^\eps$, 
$f$ disagrees with all of $\supp{\P} \setminus \set{f}$
on all but a fraction of at most $\frac{d}{|\F|} \cdot |\F|^\eps$ of the points.\\ 

Let $\P'$ be a neighbor of $\P$: $\P' \cap \P = \ell$. 
If there exists $x\in \ell$ such that $\forall g\in \supp{\P} \setminus \set{f}\colon g(x) \neq f(x)$,
then (\ref{eq: weak local-to-global}) implies $\P' \in S$. 
This happens with a probability of at least 
$1- d \cdot |\F|^{\eps-1}$, so 
$\Phi(S) \le d \cdot |\F|^{\eps-1}$.
 Applying \Cref{fact: expand-PvP}
 (on the expansion of the $\pvp$ graphs): 
\[d \cdot |\F|^{\eps-1} \ge 1 - \frac{3}{|\F|} - \frac{|S|}{|\PP|}\]
\[(d + 3) \cdot |\F|^{\eps-1} \ge 1 - \frac{|S|}{|\PP|}\] 
\end{proof}

\begin{lemma}[$\pvp$ local to global]
\label{lemma: consistency-lemma}
There exists an absolute constant $\eps > 0$ such that the following holds.
Let $\A$ be a super-assignment that assigns an integer to each $\P \in \PP$ and $f \in \P_{\le d}$.
Assume $d \le |\F|^\eps$ and $|\supp{\P}| \le |\F|^\eps$ for every $\P \in \PP$.

If $\A$ satisfies the weak Plane-vs-Plane constraints, 
then there exist integers $a_1, \dots, a_k \in \Z$, $k \le |\F|^\eps$, and global degree-$d$ functions $g_1, \dots, g_k \colon \F^t \to \F$ such that 
\[
\A = a_1 \cdot \ang{ g_1 } + \dots + a_k \cdot \ang{ g_k }.
\] 
\end{lemma}

\begin{proof}
The proof consists of three stages:
\begin{enumerate}
    \item First, find a small list of degree-$d$ polynomials $g_1, \dots, g_k\colon \F^t \to \F$, so that almost all planes may be explained via restrictions of $\{g_i\}$. 
    To do so, use the consistency of the Plane-vs-Plane test \cite{Raz1997}. 
    \item Then, use \Cref{fact: expand-PvP}, to show that a large part of $\PP$ correspond super-assignment $\A' = a_1\cdot \ang{g_1} + \dots + a_k \cdot \ang{g_k}$. 
    \item Finally, apply \Cref{non-zero-short} on $\A - \A'$, showing $\A - \A' = \vec 0$.
\end{enumerate}

\paragraph{Identifying consistent planes:}
A pair $(\P, \ell \in \P)$ is \emph{``good"}, if the assignment to $\P$ is nontrivial, and does not collide on $\ell$. 
Formally, $\supp{\P} \neq \emptyset $ and $\forall f \neq g \in \supp{\P}\colon {f}_{|\ell} \neq {g}_{|\ell}$. 
An edge $\{\P_1, \P_2\}$ is good, if 
$(\P_1, \P_1 \cap \P_2)$ and 
$(\P_2, \P_1 \cap \P_2)$ are good. Finally, a plane $\P \in \PP$ is good if at least half the incident edges are good. \\

Observe that $\Cref{non-zero-short}$ states almost all the planes have nontrivial support. 
Moreover, the Schwartz-Zippel \Cref{lem: schwartz-zippel}, combined with $|\supp{\P}| \le |\F|^\eps$, implies that for almost all $\ell \subseteq \P$, the pair $(\P, \ell)$ is good.
From the union bound, $\{\P_1, \P_2\} \in E_\pvp$ is almost always good --- and therefore nearly all the planes are good. 

\paragraph{Probabilistic setting:}
For every $\P \in \PP$, we sample $T[\P] \colon \P_{\le d} \to \F$. If $\supp{\P} \neq \emptyset$, we uniformly sample from $\supp{\P}$. Otherwise, we sample a constant $T[\P] \in \P_{\le 0}$. \\ 

Let $\P\in\PP$ be a good plane and $\P \cap \P' = \ell$ be a good edge.
The support does not collide on $\ell$
so the linear constraints (\ref{eq: weak local-to-global}) imply a non-negligible agreement between the functions:  
\[\forall f\in \supp{\P}, \forall x\in\ell \colon \exists g \in \supp{\P'}\colon f(x) = g(x)\]
The support of each plane is bounded by $|\F|^\eps$, so each $f\in \supp{\P}$ agrees with at least one $g \in \supp{\P'}$, on at least $\frac{1}{|\F|^\eps}$ of the points in $\ell$. 
From the Schwartz-Zippel \Cref{lem: schwartz-zippel}, for every $f\in \supp{\P}$ there exists at least one $g \in \supp{\P'}$, such that ${f_{|\ell}(x) = g_{|\ell}(x)}$. 
The support's size is bounded, so for every $f\in \P_{\le d}$ we have 
\[\Prob{T}{f_{|\P \cap \P'} = T[\P']_{|\P \cap \P'}} \ge |\F|^{-\eps}.\]
We are interested in avoiding tables $T$ 
with an abnormally low success probability of the PvP test
over edges incident to $\P$.
Luckily, this only happens in a negligible fraction of the tables:  
\[\Prob{T}{\Prob{\P' \cap \P = \ell}{T[\P]_{|\ell} = T[\P']_{|\ell}} < \frac{1}{4}|\F|^{-\eps}} = 
\mathop\E_{f\in \P_{\le d}}\left[\Prob{T}{\frac{1}{d_\pvp} \sum_{\P' \cap \P = \ell}{1_{g_{|\ell} = T[\P']_{|\ell}}} < \frac{1}{4}|\F|^{-\eps} \biggm| T[\P] = f}\right] \le 
\]
$\P$ is a good plane so at least $\half$ of the edges are good. 
In addition, the indicators in the summation are independent and equal to $1$ with a probability of at least $\abs{\F}^{-\eps}$, 
so we apply the well-known \emph{Chernoff bound}: 
\[
\mathop\E_{f\in \P_{\le d}}\left[\Prob{T}{\frac{1}{2\#(\text{good edges})}\sum_{\text{good } \P' \cap \P = \ell}{1_{f_{|\ell} = T[\P']_{|\ell}}} < \frac{1}{4}|\F|^{-\eps}\biggm| T[\P] = f}\right] 
\le 
\ps{\sqrt{\frac{2}{e}}}^{\#(\text{good edges})/|\F|^\eps}
\]

The number of good edges adjacent to $\P$ is much larger than $|\F|^\eps$, so the probability is exponentially small.
From now on, we ignore it and assume that for every good plane $\P \in \PP$: 
\[\Prob{\P' \cap \P = \ell}{T[\P]_{|\ell} = T[\P']_{|\ell}} \ge \frac{1}{4}|\F|^{-\eps}.\]

\paragraph{Global list decoding:} 
We use \cite{Raz1997} (\ref{thm:Plane-vs-Plane}) on $T$. 
$\eps$ is sufficiently small, so $ \half|\F|^{-3\eps}$ is larger than the error term in the theorem. 
Thus, there exists $k = O(|\F|^{3\eps})$ and $f_1, \dots f_k \colon \F^t \to \F$, of degree $d$, such that: 
\[\Prob{\P_1 \cap \P_2 = \ell}{T[\P_1]_{|\ell} = T[\P_2]_{|\ell} \land \not\exists i\colon (T[\P_1] = {f_i}_{|\P_1} \land  T[\P_2] = {f_i}_{|\P_2})
} \le |\F|^{-3\eps}\]
For every good plane, the success probability is at least $\frac{1}{4}|\F|^{-\eps}$. 
Thus, the entries of at least $1 - 8|\F|^{-2\eps}$ of the good planes agree with some $f_i$ (at least half the planes are good).\\

We sample $|\F|^{1.5 \eps}$ tables and consider their list-decodings. 
Applying the union bound, almost all the good planes agree with a function in every list-decoding. 
In addition, since $|\supp{\P}| \le |\F|^\eps$, for almost all the planes, all the functions in $\supp{\P}$ were used in at least one table. 
Denote by $f_1, \dots, f_{k'}$ the concatenation of all the list-decodings, $k' = O(|\F|^{4.5\eps})$. 
Combining the previous claims, for almost all the planes $\P \in \PP$:  
\[\forall g\in \supp{\P} \colon \exists 1 \le i \le k' \colon {f_i}_{|\P} = g\]
And we denote the set of such planes, by $S_1 \subseteq \PP$. 
\paragraph{Relating restrictions:}
The Schwartz-Zippel \Cref{lem: schwartz-zippel} implies that two functions in $\{f_1, \dots, f_{k'}\}$ agree on a point 
with probability at most $\frac{d}{|\F|}\binom{k'}{2} = O(|\F|^{10 \eps - 1})$. 
We assume $\eps$ is sufficiently small, so $\frac{d}{|\F|}\binom{k'}{2} = o(1)$. 
Thus, for almost all planes (and points), the restrictions of $\{f_i\}_{i=1}^{k'}$ are pairwise distinct.
Let $S_2 \subseteq S_1$ be the subset of those planes in $S_1$. 
Observe that $S_2$ also contains most of the planes. \\

The assignment for every plane in $S_2$ could be uniquely described with a super-assignment over $f_1, \dots, f_{k'}$ (the coefficient of $f_i$ is $\A_\P[{f_{i}}_{|\P}]$). 
Denote $\pi \colon S_2 \to  \Z^{k'}$ the mapping between a plane in $S_2$, and the coefficients of that super-assignment.
Consider the equivalence relation over $S_2$: $\P_1 \sim \P_2 \iff \pi(\P_1) = \pi(\P_2)$. 
Next, we will show that few edges cross between different equivalence classes. 
Then, we deduce one class has to contain almost all the planes.\\ 

For a plane $\P_1 \in S_2$ and an edge $\{\P_1, \P_2\} \in E_\pvp$, we can classify the edge into three types:  
\begin{enumerate}
    \item \emph{Bad} edges: $\P_2 \not\in S_2$. 
    \item \emph{Non-escaping} edges: $\P_2 \in [\P]_{/\sim}$. 
    \item \emph{Crossing} edges: $\P_2 \not\sim \P_1$.
\end{enumerate}

Since $S_2$ contains $1 - o(1)$ of the edges, and the graph is $d_\pvp$-regular, the fraction of bad edges is $o(1)$. 
If $\{\P_1, \P_2\}$ is crossing, $f_1, \dots, f_{k'}$ has to collide on $\P_1 \cap \P_2$. 
Thus, the Schwartz-Zippel \Cref{lem: schwartz-zippel} implies $o(1)$ of the edges are crossing---almost all edges are non-escaping. 
Due to \Cref{fact: expand-PvP}, at least one equivalence class has a fractional size of $1 - o(1)$.  

\paragraph{Endgame:}
Denote $\A' = a_1 \cdot \ang{f_1} + \dots + a_{k'} \cdot \ang{f_{k'}}$, the super-assignment of that equivalence class. 
The vector $\A - \A'$ assigns $\vec 0$ to the planes in that equivalence class---almost all the planes. 
From linearity, it satisfies the $\pvp$ constraints. 
The triangle inequality bounds the size of each plane's support by $|\F|^\eps + k'$. 
We choose a sufficiently small $\eps > 0$, so \Cref{non-zero-short} implies $\A - \A' = \vec 0$. \\

We remark that $k'$ is effectively bounded by $|\F|^\eps$, because for planes where $f_1, \dots, f_{k'}$ don't collide, the support's size is $|\sett{f_i}{a_i \neq 0}|$ (if $a_i$ = 0, $f_i$ is meaningless and could be ignored). 
\end{proof}

\section{Discussion and Open Problems} \label{sec: conclusions}
We extend the \emph{super-assignment} framework of \cite{Dinur1999,Dinur2002,Dinur2003} to establish improved hardness-of-approximation for both \textsf{GapSVP} and \textsf{uSVP}.
The main obstacle in porting from \textsf{CVP} to \textsf{SVP}
is excluding \emph{self-involved} super-assignments---those supported on only a small fraction of planes---whose artificially low norm does not reflect any global assignment. 
This stems from the homogeneity of \textsf{SVP}: 
a super-assignment that is zero on almost all planes cannot be ruled out naively by local constraints, unlike \textsf{CVP} where nontriviality can be enforced everywhere. 
Our first goal, therefore, is to ensure low-norm assignments \emph{disperse} across the \(\PvP\) graph.

Leveraging expansion in Grassmann graphs---a generalization of the
plane-vs-plane graph---entered the PCP toolkit following its pivotal role in resolving the 2-to-2 Games Conjecture \cite{Khot2017, Dinur2017a,Dinur2018a, Khot2023};
see also Minzer’s thesis~\cite{Minzer2022}. 
These results yield \emph{structure theorems} for small sets of planes/subspaces:
unless such a set is extremely expanding, it must exhibit a rigid structure.
This machinery has since been extended, leading to multiple applications \cite{Kaufman2025,Minzer2023,Minzer2024_optrade}. 
In our setting, much weaker expansion properties suffice, though our approach was inspired by the new understanding of expansion in such graphs.

To reach quasi-polynomial instance size we adopt Composition–Recursion \cite{Arora1998} in the algebraic variant of \cite{Dinur1999,Dinur2003}.
Composition-Recursion creates super-polynomially many planes, causing self-involved assignments to reappear. 
We counter this by passing, at each recursion level, to an \emph{extension field} of the current field, thereby keeping the number of planes/vertices polynomial in the field size at every level.
Field extensions may be useful in further Composition–Recursion applications.

Starting from \emph{unambiguous} problems is a natural starting point for our pipeline: \svpp and $\alpha$--$\textsf{BDD}_p$ with $\alpha < \frac{1}{2}$ are themselves unique.
Valiant–Vazirani show a \emph{randomized} reduction from \textsf{SAT} to \textsf{Unambiguous–3SAT} and, more generally, a pathway from NP to Promise-UP \cite{Valiant1985}.
We remark that deterministically reducing an NP-hard problem to $\textsf{uSVP}$ would imply $\text{NP}=\text{Promise-UP}$.

\subsection{Open problems}
Several directions remain open. The central challenge is to obtain comparable (or even weaker) results in the Euclidean norm \(\ell_2\).

\begin{conj}[Toward \(\ell_2\)]
Under deterministic reductions, \textsf{GapSVP} is hard to approximate in the \(\ell_2\) norm---even the \emph{exact} \textsf{SVP} problem is not known to be NP-hard (deterministically).
\end{conj}

We conjecture that PCP-style techniques can be adapted to \(\ell_2\). Self-involved assignments of low-norm in \(\ell_2\) appear highly structured, which makes their characterization and exclusion plausible. Even a sub-exponential derandomization for \(\ell_2\) would be compelling.

As emphasized in the introduction, \textsf{uSVP} underpins many cryptographic constructions; sharpening its complexity in \(\ell_2\) is therefore interesting. 
Our matching \textsf{uSVP} bounds strengthen the case for the difficulty of unique instances, yet the known \(\ell_2\) hardness still lags.

\begin{conj}[Unique \(\ell_2\)]
There exists \(\eps>0\) such that it is NP-hard to approximate \textsf{uSVP} in the \(\ell_2\) norm within a ratio \(1+\eps\), even under randomized, sub-exponential-time reductions.
\end{conj}

\subsection*{Acknowledgments} 
We thank Dor Minzer, Itamar Rot, Beata Kubis
and Yonatan Pogrebinsky for helpful discussions and comments.

\bibliography{references.bib}

\section{Appendix A} \label{Appendix: A}
To make the paper self-contained, we follow an analysis of Kaufman and Minzer \cite{Kaufman2022}, proving \Cref{fact: expand-PvP}. 
We assume basic familiarity with eigenvalues, characters, and Cayley graphs.
If needed, the survey of Hoory, Linial, and Wigderson \cite{Hoory2006} contains all the necessary background and far more. 

The proof consists of the following steps. 
First, it is possible to move from $G_\pvp$ to a certain Cayley graph (Actually, as seen in \cite{Kaufman2022}, from a more general case known as the Affine Grassmann graph). 
The eigenvalues of that Cayley graph are easy to compute. 
Once the eigenvalues are computed, the well-known \emph{Expander Mixing Lemma} implies \Cref{fact: expand-PvP}. 

\subsection{The Cayley graph}
Starting from the $\pvp$ graph over $\F_q^t$, $t > 2$, 
the vertices of our Cayley graph are the triplets $(  s,   x_1,   x_2) \in (\F_q^t)^3$. 
Each vertex $(s, x_1, x_2)$ has the following edges, described via a randomized process: 

\begin{enumerate}
    \item Uniformly sample $y \in \F_q^t$ and $\alpha, \beta, \gamma \in \F_q$. 
    \item Move to $(s + \alpha y, x_1 + \beta y, x_2 + \gamma y)$.
\end{enumerate}
\newcommand{\Mcay}{M^\text{Cay}}
Denote the transition/random-walk matrix by $\Mcay$. 
Observe that $\Mcay$ is symmetric, 
as starting from $(s, x_1, x_2)$ and using $y, \alpha, \beta, \gamma$ 
is equivalent to starting from $(s + \alpha y, x_1 + \beta y, x_2 + \gamma y)$ and using $-y, \alpha, \beta, \gamma$. 

The connection to $G_\pvp$ arises from associating each $(  s,   x_1,   x_2)$ with an affine subspace $  s + \Span(  x_1,   x_2)$. 
Given a set $S \subseteq \PP(\F_q^t)$, denote: 
\[S^* \defeq \sett{(  s,   x_1,   x_2) \in (\F_q^t)^3}{  s + \Span(  x_1,   x_2) \in S}\]
It turns out that the expansion of $S$ and $S^*$ are closely related. 
Thus, since the eigenvalues of Cayley graphs can be more easily calculated, 
in the subsequent sections we prove expansion in the Cayley graph.
The following claim allows us to derive almost the same results for $G_\pvp$. 

\begin{claim} \label{appAclaim: sameEXP}
    $\Phi(S) \ge \Phi(S^*) - \frac{1}{q} - \frac{1}{q^2}$
\end{claim}

\begin{proof}
Suppose $(s, x_1, x_2) \in S^*$
and sample a random neighbor $(s + \alpha y, x_1 + \beta y, x_2 + \gamma y)$.
There are few cases: 
\begin{enumerate}
    \item $\beta, \gamma = 0$, which happens with probability$\frac{1}{q^2}$. 
    \item $y \in \Span\ps{x_1, x_2}$, with probability $q^{2-t}$. 
    As $t = 3$, we receive $\frac{1}{q}$.  
    \item $\dim\ps{\Span(x_1 + \beta y, x_2 + \gamma y)} < 2$.
    This case is contained in $y \in \Span(x_1, x_2)$, as 
    a linear combination 
    $c_1(x_1 + \beta y) + c_2(x_2 + \gamma y)$ 
    could be rewritten as 
    $c_1 \cdot x_1 + c_2 \cdot x_2 = (\beta c_1 + \gamma c_2) y$. 
    
    \item Otherwise, $(s + \alpha y) + \Span\ps{x_1 + \beta y, x_2 + \gamma y}$ is an affine plane intersecting with $s + \Span(x_1, x_2)$ on a line. 
    Moreover, it is distributed uniformly across such planes.   
    The proof is elementary but slightly technical, and is written in the next paragraphs. 
\end{enumerate}
First, let us calculate the intersection. 
As $\beta \neq 0 \lor \gamma \neq 0$, either

\[s + \alpha y + \alpha \cdot \beta^{-1} (x_1 + \beta y) = 
s + \alpha \cdot \beta^{-1} x_1
\in s + \Span(x_1, x_2)\]
\[\text{\textbf{or} } s + \alpha y + \alpha \cdot \gamma^{-1} (x_2 + \gamma y) = 
s + \alpha \cdot \gamma^{-1} x_2
\in \Span(x_1, x_2)\]
In addition, as $y\not\in \Span(x_1, x_2)$, the planes differ. 
It's not hard to see that the intersection is a line with a gradient of: 
\[\gamma(x_1 + \beta y) - \beta(x_2 + \gamma y) = \gamma x_1 - \beta x_2\]
Fixing $\beta, \gamma$, suppose w.l.o.g $\beta \neq 0$.
The point 
$s + \alpha \cdot \beta^{-1} x_1$
is inside the intersection, so the line is: 
\[\ell(t) = s + \alpha \cdot \beta^{-1} x_1 + t\cdot \ps{\gamma x_1 - \beta x_2}\]
For different values of $\alpha$, the lines are parallel.
Thus, when sampling $\alpha, \beta, \gamma$ (independent of $y$), the intersection distributes uniformly over the lines. 
Sampling $\alpha, \beta, \gamma$ and then $y$, 
has the same distribution as uniformly sampling a line in $s + \Span(x_1, x_2)$, and then adding a random third point outside the plane.

To conclude the proof of \Cref{appAclaim: sameEXP}, 
the first three cases occur with a probability of at most $\frac{1}{q} + \frac{1}{q^2}$. 
$(s + \alpha y, x_1 + \beta y, x_2 + \gamma y) \in S^*$ is equivalent to $s + \alpha y + \Span(x_1 + \beta y, x_2 + \gamma y) \in S$, 
so the fourth case corresponds to a random walk in the $G_\pvp$ graph. 
It implies that: 
\[\frac{1}{q} + \frac{1}{q^2} + \ps{1 - \frac{1}{q} - \frac{1}{q^2}} \Phi(S) \ge \Phi(S^*) \Rightarrow
\Phi(S) \ge \Phi(S^*) - \frac{1}{q} - \frac{1}{q^2} \]

\end{proof}

\subsection{Bounding the Eigenvalues}
Our next step is calculating the eigenvalues of $\Mcay$. 
Again, this was done before in \cite{Kaufman2022}. 
It is folklore that, for Cayley graphs, the characters are the eigenvectors of the graph. 
While the graph is weighted, this is still true. 
For a character $\chi_x$, denote the corresponding eigenvalue as $\lambda_x$. 

\begin{claim} \label{appAclaim: eigenvalues}
    For all $\vec 0 \neq x = (s, x_1, x_2) \in (\F_q^t)^3$, the eigenvalue $\lambda_x$ is bounded by $\abs{\lambda_x} \le \frac{1}{q}$. 
\end{claim}

\begin{proof}
A calculation of $\Mcay \chi_x$ in any coordinate shows: 
\[\lambda_x = \sum_{u \in (\F_q^t)^3} \Mcay_{\vec 0, u}\chi_x(u) = \E_{y, \alpha,\beta,\gamma}[\chi_x(\alpha y, \beta y, \gamma y)]\]
Reordering the term inside, we receive the following: 
\[\E_{y, \alpha,\beta,\gamma}[\chi_x(\alpha y, \beta y, \gamma y)] = 
\E_{y, \alpha,\beta,\gamma}[\chi_{\alpha s + \beta x_1 + \gamma x_2}(y)] = \begin{cases}
    1 & \alpha s + \beta x_1 + \gamma x_2 = 0\\ 
    0 & \text{else}\\
\end{cases}\]
The first case happens with a probability of $q^{-\dim \Span(s, x_1, x_2)}$.
We assumed $x \neq \vec 0$ so $\abs{\lambda_x} \le \frac{1}{q}$. 
\end{proof}

\subsection{Concluding Fact \ref{fact: expand-PvP}}
One of the most fundamental ideas in the theory of expanders is that good expanders ``look random". 
The \emph{Expander Mixing lemma} \cite{Alon1988}
states that for every $d$-regular graph $G = (V, E)$ and a subset $S \subseteq V$, 
if $G$ has a good spectral expansion, 
the number of edges with endpoints in $S$ and $V$ (crossing the partition), 
is roughly $d \cdot |S|\frac{|V| - |S|}{|V|}$. 
Equivalently, $\Phi(S) \approx 1 - \frac{|S|}{|V|}$. 
While our Cayley graph is weighted, the classical proof still works. 

\begin{lemma}[Expander mixing lemma --- weighted] \label{appA: exmix}
    Let $G$ be a weighted graph on vertices $[n]$ with a \emph{symmetric} random-walk matrix $W \in [0, 1]^{n\times n}$. 
    Denote the eigenvalues $\lambda_1 \ge \dots \ge \lambda_n$.
    For $\lambda > 0$, assume $\max_{i \neq 1} |\lambda_i| \le \lambda$.
    For every $S \subset [n]$:  
    \[\abs{\Phi(S) - 1 + \frac{|S|}{n}} \le \lambda\]
\end{lemma}

Proving \Cref{fact: expand-PvP} is immediate. 
\Cref{appA: exmix} states that for every $S \subset V_\pvp$ in the $\pvp$ graph
and a corresponding $S^* \subseteq (\F_q^t)^3$ in the Cayley graph: 
\[\Phi(S) + \frac{1}{q} + \frac{1}{q^2} \ge \Phi(S^*) \ge 1 - \frac{|S^*|}{q^{3t}} - \frac{1}{q} \Rightarrow
\Phi(S) \ge 1 - \frac{2}{q} - \frac{1}{q^2} -  \frac{|S|^*}{q^{3t}}
\]
Observe that:
\[\frac{|S|}{|V_\pvp|} = \frac{|S^*|}{|\sett{(s, x_1, x_2) \in (\F_q^t)^3}{\dim\Span(x_1, x_2) = 2}|} \ge \frac{|S^*|}{q^{3t}}\]
So a loose bound on the expansion is (\Cref{fact: expand-PvP}): 
\[
\Phi(S) \ge 1  -  \frac{|S|}{|V_\pvp|} - \frac{3}{q}.
\]

\end{document}